\documentclass[a4paper,10pt]{article}

\usepackage[utf8]{inputenc}
\usepackage[english]{babel}
\usepackage[babel]{csquotes}
\usepackage{amsthm, amsmath, amssymb,amsfonts,latexsym}
\usepackage{mathtools} 
\usepackage{enumerate}
\usepackage{tikz}
\usetikzlibrary{automata,positioning}
\usepackage{fullpage}
\usepackage{cases} 
\usepackage{listliketab}
\usepackage{bbm}
\usepackage{varioref}

\numberwithin{equation}{section}

\usepackage{array}					
\newcolumntype{C}{>{\centering\arraybackslash$}p{1cm}<{$}}
\newcolumntype{D}{>{\centering\arraybackslash$}p{1.7cm}<{$}}

\RequirePackage{natbib}
\RequirePackage[colorlinks,citecolor=blue,urlcolor=blue]{hyperref}
\newcommand{\pr}[1]{\left(#1\right)}
\newcommand{\prbig}[1]{\Big(#1\Big)}
\newcommand{\prq}[1]{\left[#1\right]}
\newcommand{\norm}[1]{\left\|#1\right\|}
\newcommand{\prg}[1]{\left\{#1\right\}}
\newcommand{\abs}[1]{\left|#1\right|}
\newcommand{\Sp}{\mathcal{X}}
\newcommand{\C}{\mathcal{C}}
\newcommand{\R}{\mathcal{R}}
\newcommand{\K}{\mathcal{K}}
\newcommand{\is}{H}
\newcommand{\In}{\mathcal{V}}
\newcommand{\Ini}{\mathcal{V}_i}
\newcommand{\ii}{V}
\newcommand{\iii}{V_i}
\newcommand{\SR}{\mathcal{U}}
\newcommand{\sr}{U}
\newcommand{\FP}{\mathcal{W}}
\newcommand{\fp}{W}
\newcommand{\La}{\mathcal{L}}
\newcommand{\NN}{\mathbb{N}}
\newcommand{\RR}{\mathbb{R}}
\newcommand{\diag}{\mathrm{diag}}
\newcommand{\prcon}[2]{\pr{#1\left|#2\right.}}
\newcommand{\prconf}[2]{\pr{\left.#1\right|#2}}
\newcommand{\Econ}[2]{E\prq{#1\left|#2\right.}}
\newcommand{\Econf}[2]{E\prq{\left.#1\right|#2}}
\newcommand{\PP}{\mathbb{P}}
\newcommand{\out}[1]{d_{#1}}
\newcommand{\outx}[1]{d^x_{#1}}
\newcommand{\outN}[1]{d^N_{#1}}
\newcommand{\outxN}[1]{d^{N,x}_{#1}}
\newcommand{\outaxN}[1]{d^{N,N^\alpha x}_{#1}}
\newcommand{\df}{=}

\title{Elimination of Intermediate Species in Multiscale Stochastic Reaction Networks}
\author{Daniele Cappelletti\footnotemark[1] \and Carsten Wiuf\footnotemark[1]}
\date{}

\newtheorem{thm}{Theorem}[section]
\newtheorem{cor}[thm]{Corollary}
\newtheorem{lem}[thm]{Lemma}
\newtheorem{prop}[thm]{Proposition}

\theoremstyle{remark}
\newtheorem{rem}{Remark}[section]
\newtheorem{ex}{Example}[section]

\theoremstyle{definition}
\newtheorem{defi}{Definition}[section]
\newtheorem{ass}{Assumption}

\begin{document}

 \footnotetext[1]{Department of Mathematical Sciences, University of Copenhagen, Copenhagen, Denmark. The authors are supported by funding from the Carlsberg Foundation, the Lundbeck Foundation, the Danish Cancer Society the Danish Research Councils.}

 \tikzset{every node/.style={auto}}
 \tikzset{every state/.style={rectangle, minimum size=0pt, draw=none, font=\normalsize}}
 \tikzset{bend angle=20}

 \maketitle
 
 \begin{abstract}
 We study networks of biochemical reactions modelled by continuous-time Markov processes. Such networks typically contain many molecular species and reactions and are hard to study analytically as well as by simulation. Particularly, we are interested in reaction networks with intermediate species such as the substrate-enzyme complex in the \mbox{Michaelis-Menten} mechanism. Such species are virtually in all real-world networks, they are typically short-lived, degraded at a fast rate and hard to observe experimentally. 
 
 We provide conditions under which the Markov process of a multiscale reaction network with intermediate species is approximated by the Markov process of a simpler reduced reaction network without intermediate species. We do so by embedding the Markov processes into a one-parameter family of processes, where reaction rates and species abundances are scaled in the parameter. Further, we show that there are close links between these stochastic models and deterministic  ODE models of the same networks.
 \end{abstract}

  \section{Introduction}
 Reliable mathematical models of biochemical reaction networks are of great interest for the analysis of experimental data and theoretical biochemistry. Such models can provide qualitative information on biochemical systems as well as provide means to simulate networks and to estimate unknown parameters.    The classical stochastic model of a reaction network is a continuous-time Markov process, where the states are configurations of species numbers  and the transitions are changes caused by reactions. We refer to this Markov process as a \emph{stochastic reaction network (SRN)}.
 Unfortunately the set of reactions and chemical species is often very large, and the related Markov process is too complicated to be studied analytically or by  modern computers. Thus, the necessity of simplifying the full model arises. Perhaps the first result in this direction is due to \cite{kurtz:classical_scaling}, where a deterministic weak limit for stochastic reaction networks  is obtained \citep[see also][]{kurtz:strong}. More recently, in \cite{ball:asymptotic,kurtz:rescale,popovich:rescale}, similar asymptotic results have been obtained under more general scaling conditions than those applied in \cite{kurtz:classical_scaling, kurtz:strong}. Here the limit might have stochastic as well as deterministic components, and the limit network might consist of simplified reactions with fewer species.  In this context the concept of \emph{model reduction} arises naturally.
 
  A famous and well studied example of a biochemical system  is the Michaelis-Menten mechanism for enzyme kinetics \citep{cornish:enz-kinetics,kurtz:rescale,hardin:simplified_enzyme_kinetics,grima:ssLNA,rao:stochastic}. It is described by the reactions
 \begin{center}
  \begin{tikzpicture}
   \node[state] (E+R) at (1,0) {$E+R$};
   \node[state] (\is) at (3,0) {$\is$};
   \node[state] (E+P) at (5,0) {$E+P$};
   \node[state] (R) at (1.4,0) {$\phantom{R}$};
   \path[->] 
    (R) edge[bend left] node{} (\is)
    (\is)   edge[bend left] node{} (R)
    (\is)   edge            node{} (E+P);
  \end{tikzpicture}
 \end{center}
 where $E$ denotes an enzyme, $R$ a reacting substrate and $P$ a product. $\is$ is an intermediate, or transient, species formed by $E$ and $R$, and it is usually unstable. Whenever a reaction occurs, say $E+R\to \is$, then the number of molecules changes accordingly, that is, the numbers of $E$ and $R$ molecules are each reduced by one, while the number of $\is$ molecules is increased by one.
 
 If we assume that at least one of the reactions $\is\rightarrow E+R$ and $\is\rightarrow E+P$ is so fast that a produced molecule of $H$ is quickly degraded before any other reaction takes place (that is, at any time at most one molecule of $H$ is  present), then it seems reasonable that the Markov process could be approximated by a simpler Markov process, corresponding to the \emph{reduced} reaction network
  \begin{center}
  \begin{tikzpicture} 
   \node[state] (E+R) at (0,0) {$E+R$}; 
   \node[state] (E+P) at (2,0) {$E+P$};
   \path[->] 
    (E+R) edge node {} (E+P);
  \end{tikzpicture}
 \end{center}
where the reaction rate is determined from the original reaction rates. Intuitively, the rate is the rate of $E+R\to H$ multiplied by the probability that the reaction $\is\to E+P$ occurs instead of $\is\to E+R$. Under this reduction the number of enzyme molecules $E$ becomes constant. In essence, we are here dealing with \emph{time-scale separation}, in addition to \emph{species elimination} and \emph{dimensionality reduction} (both in terms of the number of reactions as well as the number of species). 

Another, perhaps more interesting  example, is the following reaction network:
 
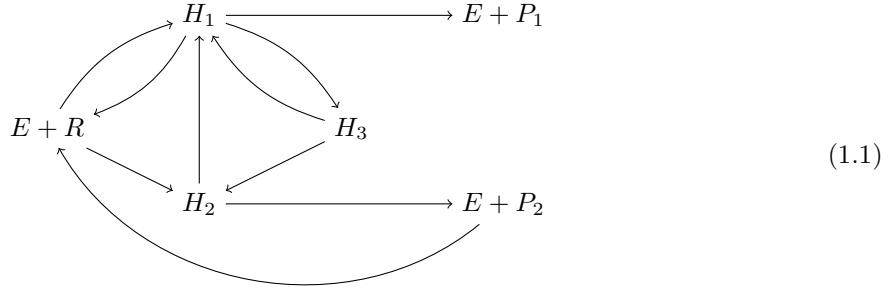
\begin{equation}\label{eq:example_without_rates_big}
  \begin{tikzpicture}[baseline={(current bounding box.center)}]
   \node[state] (E+R)  at (1,0)    {$E+R$};
   \node[state] (\is1)   at (3,1.5)  {$\is_1$};
   \node[state] (\is2)   at (3,-1)   {$\is_2$};
   \node[state] (\is3)   at (5,0)    {$\is_3$};
   \node[state] (E+P1) at (7,1.5)  {$E+P_1$};
   \node[state] (E+P2) at (7,-1)   {$E+P_2$};
   \path[->]
    (E+R)  edge[bend left] node {} (\is1)
           edge            node {} (\is2)
    (\is1)   edge[bend left] node {} (E+R)
           edge            node {} (E+P1)
           edge[bend left] node {} (\is3)
    (\is2)   edge            node {} (E+P2)
           edge            node {} (\is1)
    (\is3)   edge[bend left] node {} (\is1)
           edge            node {} (\is2)
    (E+P2) edge[bend left=50] node {} (E+R);
  \end{tikzpicture}
 \end{equation}
 It describes the catalytic transformation of a species $R$ into the species  $P_1$ or $P_2$, through a chain of intermediate steps, denoted by the species $\is_1$, $\is_2$ and $\is_3$. Whenever the reaction $E+R\to\is_1$ occurs, a sequence of reactions between intermediate species will take place (for example, $\is_1\to\is_3\to\is_1$) before a final complex is produced, such as $E+P_1$. If the time spent in intermediate states is small, we might approximate the reaction paths proceeding through the formation and quick degradation of intermediate species by direct reactions. In other words, it is reasonable to contract reaction paths passing through any intermediate species to  obtain
  
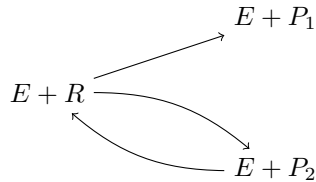
\begin{equation*}
  \begin{tikzpicture}[baseline={(current bounding box.center)}]
   \node[state] (E+R)   at (1,0)  {$E+R$};
   \node[state] (E+P1)  at (4,1)  {$E+P_1$};
   \node[state] (E+P2)  at (4,-1) {$E+P_2$};
   \path[->]
    (E+R)  edge            node{} (E+P1)
           edge[bend left] node{} (E+P2)
    (E+P2) edge[bend left] node{} (E+R);
  \end{tikzpicture}
 \end{equation*}
for a suitable choice of reaction rates.  Note that there is an infinite number of such reaction paths. We will provide conditions that guarantee that the original SRN can be well approximated, in a certain sense, by the reduced SRN, or more accurately, that the Markov process describing the original system  is well approximated by the Markov process of the reduced system.

  For this aim, we introduce a family of  \emph{kinetics} (reaction rates) indexed by a parameter $N$ and study the  relationship between the original and the reduced SRNs as $N\to\infty$.  The analysis builds on the previous work \cite{feliu:intermediates} \citep[see also][]{feliu:elimination,feliu:elimination_bio}, as well as on \cite{ball:asymptotic,kurtz:rescale,popovich:rescale}.  In \cite{feliu:intermediates}, a mathematical framework is developed for the elimination of intermediate species in deterministically modelled reaction networks, using ODEs.  Properties of the steady states in the original ODE system are related to similar properties of the steady states in the reduced ODE system by means of a formal relationship between the original and the reduced network. Here we are not concerned about the steady states  nor about the equilibrium distributions of SRNs, but about the trajectories of SRNs up to a finite fixed time $T>0$. Our aim is to approximate the dynamics of the original system with intermediate species by means of the dynamics of a simplified model, where intermediate species are eliminated. Though we arrive at our  reduced model through a different route than \cite{feliu:intermediates}, we will show that there are close links to ODE models  and that our reduced network in fact is that of \cite{feliu:intermediates}. 
 
  We will study different types of convergence of  stochastic processes associated with  SRNs as $N\to\infty$. The limit is taken assuming that the consumption rates (at least some of them) of the intermediate species approach infinity according to $N$. Also the molecular abundances might be scaled in powers of $N$ in the spirit of the \emph{multiscale analysis} performed in \cite{ball:asymptotic,kurtz:rescale,popovich:rescale}. These papers deal with various forms of model reduction. However, the elimination of intermediate species we aim to achieve is not possible in these settings. On the other hand, our approximating model might in some cases be further reduced by techniques developed in these papers, hence our approach might be considered  complementary to theirs.
 
 \section{Preliminaries and definitions}
 
 The space of real (natural) vectors with entries indexed by a finite set $A$ is denoted by $\RR^A$ ($\NN^A$), and for any vector $v\in\RR^A$ ($\NN^A$), we denote the entry corresponding to $a\in A$ by $v(a)$. Moreover, for any two vectors $v,w\in\RR^A$ ($\NN^A$) we write $v> w$ if the inequality holds component-wise. Furthermore, $\abs{v}$ denotes the usual Euclidean norm of $v$. Finally, if $A$ is a finite set, we let $\# A$ denote the cardinality of $A$. Given two real numbers $x,y$, we will often use the notation $x\vee y$ or $x\wedge y$ to denote the maximum and the minimum of $x$ and $y$, respectively.
 
 A reaction network consists of a set of species $\Sp$, a set of complexes $\C$, and a set of reactions $\R$. Formally, $\Sp$ is a finite non-empty set $\{S_1,S_2,\dots,S_n\}$, $\C=\{y_1,y_2,\dots,y_m\}$ is a non-empty set of non-negative linear combinations of elements of $\Sp$ and $\R$ is a finite non-empty subset of $\C\times\C$, such that $(y_i,y_i)\not\in \R$ for all $i$. We identify $\Sp$ and $\C$ with finite subsets of $\NN^{\Sp}$. If $(y_i,y_j)\in\R$ we write $y_i\rightarrow y_j$ and we say that $y_i$ is the \emph{reactant} and $y_j$ is the \emph{product}. Throughout the paper we will denote an object $O$ associated with a reaction $r\colon y_i\rightarrow y_j$ by $O_r$ or $O_{i j}$ indifferently. Furthermore, for each reaction $r\colon y_i\rightarrow y_j\in\R$, we define the \emph{reaction vector}
 $$\xi_r\df y_j-y_i.$$
 For further background on reaction networks, see \cite{erdi:mathematical_models, anderson:design}.

 A complex $y\in\C$ is given as $y=(y(S_1),\dots,y(S_n))$ and $y(S)$ is called the  \emph{stoichiometric coefficient} of the species $S$ in $y$.
 Furthermore, we define the \emph{support} of $y$ as the set of species $S$ such that $y(S)>0$, in which case  we write $S\in y$. Moreover, define $\C_S$ as the complexes whose support contains $S$ and $\R_S$ as the reactions in $\R$ that change the counts of $S$:
 \begin{align}
  \label{eq:C_S}
  \C_S&\df\prg{y\in\C\,\colon\,S\in y},\\
  \label{eq:R_S}
  \R_S&\df\prg{r\in\R\,\colon\,\xi_{r}(S)\neq 0}.
 \end{align}
 
 Finally, we define a \emph{kinetics} $\K$ as a set of functions indexed by $\R$ of the form
 $$\begin{array}{rrcl}
   \lambda_r\colon & \NN_{\geq 0}^{\Sp} & \rightarrow & \RR_{\geq 0} \\
                   & x                  & \mapsto     & \lambda_r(x).
  \end{array}$$
 Intuitively, $\lambda_r$ is the rate by which  reaction $r$ occurs and it will be referred to as the \emph{reaction rate}. We allow reaction rates to be constantly $0$, in which case the corresponding reaction could be removed from the network.
  
 A reaction network equipped with a kinetics can be modelled as a continuous-time Markov process $X_\cdot$ on $\NN^{\Sp}$, where $X_t(S)$ is the number of molecules of the species $S$ at time $t$. Taken together with $\K$ and $X_\cdot$, a reaction network is called a \emph{stochastic reaction network} (\emph{SRN}). The state of $X_\cdot$ changes whenever a reaction takes place, for example, if the reaction $r$ occurs at time $t^*$ the new state is
 $$X_{t^*}=X_{t^*-}+\xi_r.$$
 The kinetics $\K$ represents the transition rates for the process $X_\cdot$, such that 
 \begin{equation}\label{eq:general_form_stochastic_reaction_network}
  X_t=X_0+\sum_{r\in\R}\xi_r Y_r\pr{\int_0^t\lambda_r(X_s)ds},
 \end{equation}
 with $Y_r(\cdot)$ independent and identically distributed unit-rate Poisson processes \citep{kurtz:strong}. The random variable $Y_r\pr{\int_0^t\lambda_r(X_s)ds}$ counts how many times the reaction $r$ has occurred up to time $t$. This stochastic model is typically chosen if the number of reactant molecules is low, so that the behaviour of the system is similar to the evolution of a jump process. Changes occur only in a discrete set of time points and it is uncertain which reaction will take place next. 
 
 A typical choice of kinetics is \emph{mass-action kinetics}, where the reaction rate of $r\colon y_i\to y_j$ is given by
 $$\lambda_{r}(x)=k_{r}\prod_{S\in y_i}\frac{x(S)!}{(x(S)-y_i(S))!}\mathbbm{1}_{\{x(S)\geq y_i(S)\}},$$
 and $k_{r}$ are non-negative real numbers, called \emph{rate constants}. We usually express this as $y_i\xrightarrow{k_{r}}y_j.$ Note that the reaction rates are proportional to the number of ordered subsets of molecules that can give rise to an occurrence of the reaction. This choice of kinetics is natural if we assume  the system is well stirred.
 
 To define a reduced reaction network we introduce the concept of an intermediate species \citep{feliu:intermediates}.
 \begin{defi}\label{def:intermediate_species}
  Let $(\Sp,\C,\R)$ be a reaction network and $\In\subset\Sp$. We say that the species in $\In$ are \emph{intermediate species} (or simply \emph{intermediates}) if the following conditions hold:
  \begin{itemize}
   \item for each $\is\in\In$ and $y\in\C$, if $\is$ is in the support of $y$, then $y=\is$. This implies that $\In\subset\C$.
   \item for each $\is\in\In$, there is a directed path of complexes such that 
   $$y_i\rightarrow \is_{\ell_1}\rightarrow\dots\rightarrow \is\rightarrow\dots\rightarrow \is_{\ell_k}\rightarrow y_j$$
   for some complexes $y_i, y_j\in\C\setminus\In$ and $H_{\ell_i}\in\In$ for all $1\leq i\leq k$.
   The path
   $$\is_{\ell_1}\rightarrow\dots\rightarrow \is\rightarrow\dots\rightarrow \is_{\ell_k}$$
   is called a \emph{chain of intermediates}.
  \end{itemize}
 \end{defi}
 
According to the definition, intermediate species always appear alone and with stoichiometric coefficient one. For example, the species $\is$ in the Michaelis-Menten mechanism and the species $\is_1$, $\is_2$ and $\is_3$ in \eqref{eq:example_without_rates_big} meet Definition \ref{def:intermediate_species}. We denote by $\SR$, $\FP$ the subsets of $\C$ such that
 \begin{itemize}
  \item for all $y\in\SR$, there exists $\is\in\In$, such that $y\rightarrow \is\in\R$
  \item for all $y\in\FP$, there exists $\is\in\In$, such that $\is\rightarrow y\in\R$
 \end{itemize}
 We refer to $\SR$ and to $\FP$, respectively, as the \emph{initial reactants} and the \emph{final products}. In general the two sets can have non-empty intersection (as in Example \ref{ex:main_example_wo_N}). For any initial reactant $y_i$ we introduce the set $\Ini$ of intermediate species $\is$ such that $y_i\rightarrow \is\in\R$.
 We index the set $\In$ using the ordering of the set $\C$, such that $H_\ell=y_\ell$ for any intermediate $\is_\ell\in\In$. Further, we introduce the index sets $\sr$, $\ii$, $\iii$ and $\fp$ of $\SR$, $\In$, $\Ini$ and $\FP$, respectively, such that
 $$\SR=\prg{y_i}_{i\in\sr},\quad\In=\prg{\is_\ell}_{\ell\in\ii},\quad\Ini=\prg{\is_\ell}_{\ell\in\iii},\quad\FP=\prg{y_j}_{j\in\fp}.$$
 
 \section{The Reduced Stochastic Reaction Network}\label{sec:the_reduced_model}
 
 Let $(\Sp, \C, \R)$ be a reaction network equipped with a kinetics $\K$ and let $\In\subset\Sp$ be a set of intermediate species. 
 
 The reduced reaction network obtained from $(\Sp, \C, \R)$ is the triple 
 \begin{equation}\label{eq:reduced_network}
  (\Sp\setminus\In\,,\,\C\setminus\In\,,\,\R^*),
 \end{equation}
 where $\R^*$ consists of the reactions in $\R$ not involving intermediates and the reactions $y_i\rightarrow y_j$, where $y_j$ is obtainable from $y_i$ through a chain of intermediate species of $(\Sp,\C,\R)$, as in Definition \ref{def:intermediate_species}. Thus, the intermediate species have been eliminated from the original network by contraction of reaction paths.
 
 If $(\Sp,\C,\R)$ is equipped with a kinetics $\K$, then $(\Sp\setminus\In,\C\setminus\In,\R^*)$ inherits a kinetics $\K^*$ from $(\Sp,\C,\R)$ if certain additional conditions are fulfilled. To define $\K^*$ we first make the following assumption:
 
 \begin{ass}[Rate functions and intermediates]\label{ass:independence_of_the_dynamics_without_N}
  The consumption of the intermediate species is governed by mass-action kinetics, that is for any $\ell,\ell^\prime\in\ii$ and $j\in\fp$,
  $$\lambda_{\ell j}(x)=k_{\ell j}x(\is_\ell),\quad\text{and}\quad\lambda_{\ell\ell^\prime}(x)=k_{\ell \ell^\prime}x(\is_\ell),$$
  for some non-negative constants $k_{\ell j}$, $k_{\ell \ell^\prime}$. This condition implies that any molecule of an intermediate species will be consumed at a constant rate. Further, we assume that all other reaction rates do not depend on $\is_\ell$.
 \end{ass}
 
 Let $X_\cdot$ be the process associated with $(\Sp,\C,\R)$. We enlarge the filtration of $X_\cdot$ by the $\sigma$-algebras $\sigma_t$, such that $\sigma_t$ contains the information on the evolution up to time $t$ of every  occurrence of a molecule of an intermediate species  in the experiment. In particular, we introduce a Markov process, that describes the dynamics, or fate, of a molecule of an intermediate species. Consider the $n$-th reaction occurring in $X_\cdot$ that turns a non-intermediate complex into an intermediate species. Let this reaction be $y_i\rightarrow \is_\ell$ and assume it takes place at time $t_n$. The intermediate molecule $\is_\ell$ will eventually be transformed into a final product $y_j$. The chain of transformations leading to $y_j$ can be described by a continuous-time Markov chain $C_n(\cdot)$, starting at time $t_n$, with state space $\In \cup \FP$ and $C_n(t_n)\in\In$. The final products are treated as absorbing states for the Markov process. The transition rate matrix, which is independent on $n$, has the following block structure:
 {\renewcommand\arraystretch{1.9}
  \begin{equation}\label{eq:transition_matrix}
   Q=\left[\begin{array}{C|C}
            Q_{\ii, \ii} & Q_{\ii, \fp} \\ \hline
            0 & 0 \\
           \end{array}
        \right],
  \end{equation}}
 where
  \begin{alignat*}{2}
   q_{\ell\ell^\prime}&=k_{\ell \ell^\prime}           &&\qquad\text{for all } \ell,\ell^\prime\in\ii\text{ with }\ell\neq\ell^\prime \\
   q_{\ell j}         &=k_{\ell j}                     &&\qquad\text{for all } \ell\in\ii\text{ and } j\in\fp  \\
   q_{\ell \ell}      &=-\sum_{\ell^\prime\in\ii}k_{\ell \ell^\prime}-\sum_{j\in\fp}k_{\ell j} &&\qquad\text{for all } \ell\in\ii.
  \end{alignat*}
 We define by $\tau_n$ the time until the production of the final product, i.e.
 \begin{equation*}
  \tau_n\df\inf\prg{t\geq t_n \colon C_n(t)\in\FP}-t_n,
 \end{equation*}
 and for all $\ell\in\iii$, we define by $\pi_{\ell j}$ the probability that the final product produced is $y_j$, given that the intermediate chain started in $\is_\ell$. Namely,
 \begin{equation}\label{eq:pi}
  \pi_{\ell j}\df P\prcon{C_n(t_n+\tau_n)=y_j}{C_n(t_n)=\is_{\ell}},
 \end{equation}
 with $\pi_{\ell j}=0$ if $j\notin\fp$. Since $C_n(\cdot)$ is a finite state Markov process with absorbing states, $\tau_n$ is almost surely finite. Moreover note that $\pi_{\ell j}$ does not depend on $n$, since $Q$ does not depend on $n$. In this context, we have
 \begin{equation}\label{eq:sigmat}
 \sigma_t=\sigma\pr{X_s,C_n(s)\colon s\in[0,t], n\in\NN}.
 \end{equation}
 
 Let $\K$ be a kinetics fulfilling Assumption \ref{ass:independence_of_the_dynamics_without_N}. If we let $\lambda_{i\ell}=0$ whenever $y_i\rightarrow H_\ell\notin\R$, then  the kinetics $\K^*$ of the reduced reaction network is defined by
 \begin{equation}\label{eq:def_of_rates_Z}
  \lambda^*_{ij}(x)=\lambda_{ij}(x)+\sum_{\ell\in\iii} \pi_{\ell j}\lambda_{i\ell}(x),
 \end{equation}
 for any $y_i\rightarrow y_j\in\R^*$. Thus, the rate of a reaction originating from a chain of intermediates is the sum of the rates $\lambda_{i\ell}(\cdot)$ by which the first intermediate is produced from $y_i$ multiplied by the probability  $\pi_{\ell j}$ that the chain ends in $y_j$. To this we add $\lambda_{ij}(\cdot)$ if the reaction $y_i\rightarrow y_j$ is already in $\R$.

 Our main goal is to prove that the behaviour of $X_\cdot$, under certain conditions, is captured by the behaviour of the process associated with the reduced SRN. In the broader setting of multiscale models \citep{ball:asymptotic,kurtz:rescale,popovich:rescale}, we prove that a suitable rescaled version of $X_\cdot$ can be approximated by a similarly rescaled version of the process of the reduced SRN. We will show this by constructing a particular process $Z_\cdot$ on the same probability space as $X_\cdot$, which is distributed as the process associated with the reduced SRN, and by further proving convergence in probability of the difference between the rescaled versions of $X_\cdot$ and $Z_\cdot$ in various senses. Specifically, we are able to prove uniform punctual convergence in probability to zero as well as convergence in occupation measure (cf.\ Theorems \ref{thm:eq_true_for_bounded_rates} and \ref{thm:main_theorem}).  Under additional assumptions, we prove convergence in probability  to zero of the difference of the rescaled processes in the Skorohod topology  (cf.\ Theorems \ref{thm:weak_convergence_bounded_rates} and \ref{thm:main_theorem}).
  
 The reduced reaction network defined here is the same as the  reduced reaction network introduced in \cite{feliu:intermediates}. Moreover, the procedure to obtain the kinetics of the reduced model coincides with that in  \cite{feliu:intermediates}. We prove this in Theorem \ref{thm:model_comparison}. It is worth noting, however, that the aims of  \cite{feliu:intermediates} and this paper are very different. Indeed, we study various convergences of the stochastic processes associated with $(\Sp,\C,\R)$, while in \cite{feliu:intermediates} the reaction networks are deterministically modelled through a system of ODEs, and a relation between the  steady states of the original and the reduced models is investigated. 
 
 In \cite{feliu:intermediates}, the kinetics of the reduced reaction network is given by
 \begin{equation}\label{eq:reduced_rates_in_wiuf}
  \widetilde{\lambda}_{ij}(x)=\lambda_{ij}(x)+\sum_{\ell\in\iii} k_{\ell j}\mu_{i\ell}(x),
 \end{equation} where $\mu_{i\ell}$ is defined as follows: consider the labelled directed graph $\mathcal{G}_i^x$ with node set $\In\cup\{\star\}$ and labelled edge set given by:
 \smallskip
 \begin{equation}\label{eq:G_i^x}
  \text{\storestyleof{itemize}
 \begin{listliketab}
  \begin{tabular}{Lll}
   \textbullet & $\is_\ell\xrightarrow[\phantom{\sum_{j\in\fp} k_{\ell j}}]{k_{\ell\ell^\prime}}\is_{\ell^\prime}$ & if $k_{\ell\ell^\prime}\neq 0$ and $\ell\neq\ell^\prime$ \\
   \textbullet & $\is_\ell\xrightarrow{\sum_{j\in\fp} k_{\ell j}}\star$ & if $\displaystyle\sum_{j\in\fp} k_{\ell j}\neq 0$ \\
   \textbullet & $\hspace*{0.25cm}\star\xrightarrow[\phantom{\sum_{j\in\fp} k_{\ell j}}]{\lambda_{i \ell}(x)}\is_\ell$ & if $\lambda_{i \ell}(x)\neq0$
  \end{tabular}
 \end{listliketab}}
 \end{equation}
 We recall some notion from graph theory: let $\mathcal{G}$ be a labelled directed graph. A labelled spanning tree of $\mathcal{G}$ rooted at some node $g$ is a labelled directed graph $\zeta$ satisfying the following conditions:
 \begin{enumerate}[i)]
  \item the set of nodes of $\zeta$ coincides with the set of nodes of $\mathcal{G}$;
  \item any directed edge of $\zeta$ is a directed edge of $\mathcal{G}$, and the labels are conserved;
  \item $\zeta$ contains no cycle;
  \item for any node $g'\neq g$, in $\zeta$ there exists a directed path from $g'$ to $g$.
 \end{enumerate}
 Let $\Theta_i^x(\cdot)$ be the set of labelled spanning trees of $\mathcal{G}_i^x$ rooted at the argument, and let $w(\cdot)$ be the product of the edge labels of the tree in the argument. Then,  $\mu_{i\ell}(x)$ is defined as
 \begin{equation}\label{eq:definition_of_mu}
  \mu_{i\ell}(x)\df\frac{\sum_{\zeta\in\Theta_i^x(\is_\ell)}w(\zeta)}{\sum_{\zeta\in\Theta_i^x(\star)}w(\zeta)}.
 \end{equation}
There might be no spanning tree rooted at a given intermediate species for some $x\in\mathbb{N}^{\Sp}$. In that case, $\mu_{i\ell}(x)$ is 0. The denominator is always strictly positive as any intermediate is eventually turned into a non-intermediate (Definition \ref{def:intermediate_species}). The proof of the following result is given in Section \ref{sec:proofs}.
 
 \begin{thm}\label{thm:model_comparison}
  For all $x\in\mathbb{N}^{\Sp}$, $i\in\sr$, $j\in\fp$, we have $\lambda^*_{ij}(x)=\widetilde{\lambda}_{ij}(x)$, hence  \eqref{eq:def_of_rates_Z} and \eqref{eq:reduced_rates_in_wiuf}  coincide.
 \end{thm}
 
 Below we give an example of a reduced SRN.
 
 \begin{ex}\label{ex:main_example_wo_N}
  Consider the reaction network with intermediate species $\is_1$, $\is_2$, taken  with mass-action kinetics
  
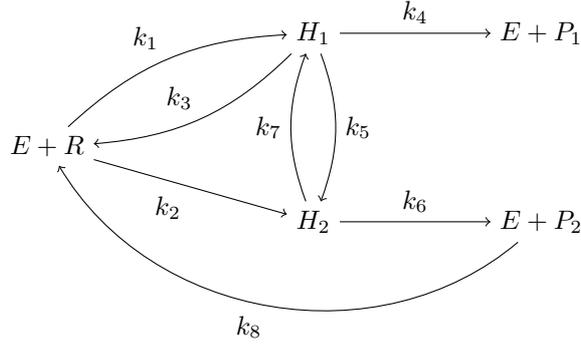
\begin{center}
   \begin{tikzpicture}[baseline={(current bounding box.center)}]
    \node[state] (E+R)  at (.5,0)    {$E+R$};
    \node[state] (\is1) at (4,1.5)  {$\is_1$};
    \node[state] (\is2) at (4,-1)   {$\is_2$};
    \node[state] (E+P1) at (7,1.5)  {$E+P_1$};
    \node[state] (E+P2) at (7,-1)   {$E+P_2$};
    \path[->]
     (E+R)  edge[bend left] node {$k_1$} (\is1)
            edge            node[swap] {$k_2$} (\is2)
     (\is1) edge[bend left] node[swap] {$k_3$}   (E+R)
            edge            node {$k_4$} (E+P1)
            edge[bend left] node {$k_5$} (\is2)
     (\is2) edge            node {$k_6$} (E+P2)
            edge[bend left] node {$k_7$} (\is1)
     (E+P2) edge[bend left=50] node {$k_8$} (E+R);
   \end{tikzpicture}
  \end{center}
  In this case there is only one initial reactant, namely $E+R$, while the final products are $E+R$, $E+P_1$ and $E+P_2$. Therefore the set of initial reactants and the set of final products have \mbox{non-empty} intersection. If we let $E+P_1=y_3$ and $E+P_2=y_4$, then,  by summing the probabilities of all possible paths from $\is_1$ to $E+P_1$, we find that
  $$\pi_{13}=\frac{k_4}{k_3+k_4+k_5}\sum_{n\in\NN}\pr{\frac{k_5}{k_3+k_4+k_5}\cdot\frac{k_7}{k_6+k_7}}^n=\frac{k_4(k_6+k_7)}{(k_3+k_4)(k_6+k_7)+k_5k_6}.$$
  Similarly, we calculate $\pi_{14}$, $\pi_{23}$ and $\pi_{24}$ and obtain
  \begin{align*}
   \pi_{14}&=\frac{k_5k_6}{(k_3+k_4)(k_6+k_7)+k_5k_6},\\
   \pi_{23}&=\frac{k_4k_7}{(k_3+k_4)(k_6+k_7)+k_5k_6},\\
   \pi_{24}&=\frac{(k_3+k_4+k_5)k_6}{(k_3+k_4)(k_6+k_7)+k_5k_6}.
  \end{align*}
  The reduced reaction network with mass-action kinetics is therefore
  
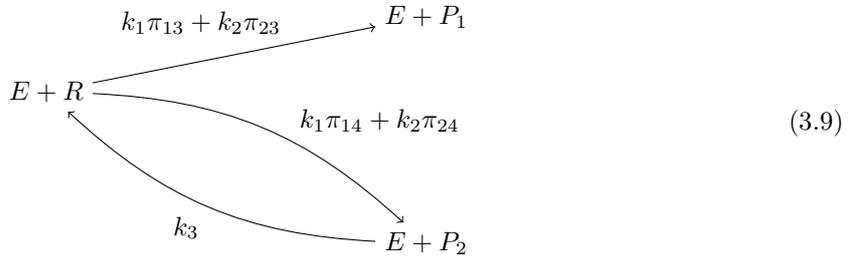
\begin{equation}\label{eq:example_reduced_network_wo_N}
   \begin{tikzpicture}[baseline={(current bounding box.center)}]
    \node[state] (E+R)   at (0,0)  {$E+R$};
    \node[state] (E+P1)  at (5,1)  {$E+P_1$};
    \node[state] (E+P2)  at (5,-2) {$E+P_2$};
    \path[->]
     (E+R)  edge            node[pos=.7]{$k_1\pi_{13}+k_2\pi_{23}$} (E+P1)
            edge[bend left] node[pos=.6]{$k_1\pi_{14}+k_2\pi_{24}$} (E+P2)
     (E+P2) edge[bend left] node{$k_3$} (E+R);
   \end{tikzpicture}
  \end{equation}
 \end{ex}
 
 \section{Results}
 
 Before formalising the setting and the assumptions, we provide some examples to motivate it. Recall Example \ref{ex:main_example_wo_N}. Intuitively, the reduced SRN behaves similarly to the original SRN if the time spent in intermediate states (states with at least one intermediate molecule being present)  is insignificant compared to the time spent in other states. Thus, it is natural to consider situations for which the reaction rates out of intermediate states are all high, though this is not what is required for our results to hold (Example \ref{ex:negative_powers_N}).
 
 Consider a reaction network $(\Sp,\C,\R)$ and a sequence of kinetics $\K^N$ indexed by $N\in\NN$. Let $X^N_\cdot$ be the process \eqref{eq:general_form_stochastic_reaction_network} associated with the kinetics $\K^N$. Generally, we will have in mind that the consumption rates of the intermediates species increase in $N$. We will consider a multiscale setting, where the species abundances also are scaled according to $N$.  Hence, we consider the asymptotic behaviour of the process $X^N_\cdot$ as $N\to\infty$, when both  species abundances and rate constants depend on $N$, similarly to what is done in \cite{ball:asymptotic,kurtz:rescale,popovich:rescale}. 
 
 To increase readability, in the examples the reaction rates depending on $N$ are simple powers of $N$ with no prefactors (e.g.\  $N^2$ rather than $kN^2$). In the results these restrictions are not assumed and more general forms of reaction rates are allowed.

  \begin{ex}\label{ex:main_example}
  Consider the SRN from Example \ref{ex:main_example_wo_N} with rate constants
  
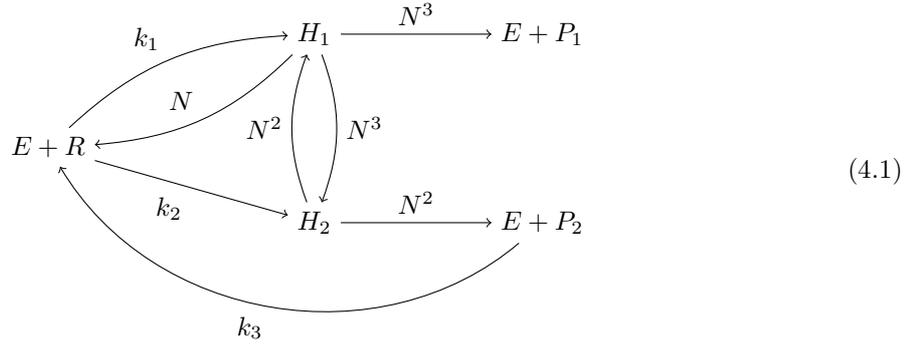
\begin{equation}\label{eq:example_with_rates_small}
   \begin{tikzpicture}[baseline={(current bounding box.center)}]
    \node[state] (E+R)  at (.5,0)   {$E+R$};
    \node[state] (\is1) at (4,1.5)  {$\is_1$};
    \node[state] (\is2) at (4,-1)   {$\is_2$};
    \node[state] (E+P1) at (7,1.5)  {$E+P_1$};
    \node[state] (E+P2) at (7,-1)   {$E+P_2$};
    \path[->]
     (E+R)  edge[bend left] node {$k_1$} (\is1)
            edge            node[swap] {$k_2$} (\is2)
     (\is1) edge[bend left] node[swap] {$N$}   (E+R)
            edge            node {$N^3$} (E+P1)
            edge[bend left] node {$N^3$} (\is2)
     (\is2) edge            node {$N^2$} (E+P2)
            edge[bend left] node {$N^2$} (\is1)
     (E+P2) edge[bend left=50] node {$k_3$} (E+R);
   \end{tikzpicture}
  \end{equation}
  The reduced SRN has reaction rates given by \eqref{eq:example_reduced_network_wo_N} with
  $$\pi^N_{13}=\frac{2N^3}{3N^3+2N},\quad\pi^N_{14}=\frac{N^3}{3N^3+2N},\quad\pi^N_{23}=\frac{N^3}{3N^3+2N},\quad\pi^N_{24}=\frac{2N^3+N}{3N^3+2N}.$$
  We assume that the molecular abundances of $R,P_1,P_2$ are of order $O(N)$, while $X^N_t(E)=O(1)$. We further assume that at time 0 there are no intermediates present, that is, $X_0(\is_1)=X_0(\is_2)=0$. The expression $O(N)$ will be made precise later, but it indicates that at a typical time $t>0$, the molecular abundances of $R,P_1,P_2$ are of the same order of magnitude as $N$. With the assumption on the abundances, the rates of the reactions $E+R\rightarrow \is_1$, $E+R\rightarrow \is_2$ and $E+P_2\rightarrow E+R$ are of order $O(N)$, while the intermediate species are consumed considerably faster. Therefore it seems reasonable that the intermediates might be eliminated from the description of the system and the dynamics described by the simpler reduced SRN in \eqref{eq:example_reduced_network_wo_N}. We will show that the dynamics of the reduced SRN approximates the dynamics of (\ref{eq:example_with_rates_small}) for $N$ large. Specifically, we will show that the difference between the two stochastic processes associated with the two networks converges to 0 in the sense of Theorems \ref{thm:eq_true_for_bounded_rates} and \ref{thm:main_theorem} for $N\rightarrow\infty$.
 \end{ex}
 
 \begin{ex}[trapped in the intermediate chain]\label{ex:time_does_not_go_to_zero}
  Consider the same reaction network as in Example \ref{ex:main_example}, but with slightly changed reaction rates.
  The reaction $\is_2\rightarrow E+P_2$ is slowed down and has rate $N$ (before $N^2$). The reaction $\is_1\rightarrow \is_2$ is accelerated and has rate $N^4$ (before $N^3$). All other rates are left unchanged. We assume as before that the molecular abundances of $R,P_1,P_2$ are of order $O(N)$, while $X^N_t(E)=O(1)$. Although the intermediate species are consumed faster than the other species (the life time of a molecule of $\is_1$ and of $\is_2$ are of order $O(1/N^4)$ and $O(1/N^2)$, respectively), it is not possible to approximate the above SRN with one of the form (\ref{eq:example_reduced_network_wo_N}), for any choice of kinetics. Indeed, it is more likely that an intermediate molecule is transformed into another intermediate molecule than into one of the two final products, $E+P_1$ and $E+P_2$. On average, an intermediate molecule will undergo the cycle of transformations $\is_1\rightarrow \is_2 \rightarrow \is_1$ $N$ times before producing a non-intermediate complex. Since the life time of a molecule of $\is_2$ is of order $O(1/N^2)$, the expected time until consumption of such a cycle of intermediates is of the order $O(1/N)$, while the rate of production of intermediate molecules is of order $O(N)$ when molecules of $E$ are present, according to the hypothesis $X^N_t(R)=O(N)$. This will result in a positive number of intermediate species being present at any fixed time $t$. Therefore, in this case, the intermediate species cannot be eliminated in the sense of this paper. 
 \end{ex}
 \begin{ex}[rescaling of time]\label{ex:rescaling_of_time}
  Consider the following SRN, which is a modified version of (\ref{eq:example_with_rates_small}). The enzyme $E$ is removed from the product complexes $E+P_1$ and $E+P_2$, and the reaction $E+P_2\rightarrow E+R$ is deleted:
  
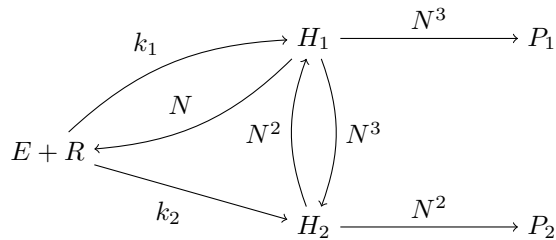
\begin{center}
   \begin{tikzpicture}
    \node[state] (E+R) at (.5,0)   {$E+R$};
    \node[state] (\is1)  at (4,1.5)  {$\is_1$};
    \node[state] (\is2)  at (4,-1)   {$\is_2$};
    \node[state] (P1)  at (7,1.5)  {$P_1$};
    \node[state] (P2)  at (7,-1)   {$P_2$};
    \path[->]
     (E+R)  edge[bend left] node {$k_1$} (\is1)
            edge            node[swap] {$k_2$} (\is2)
     (\is1) edge[bend left] node[swap] {$N$}   (E+R)
            edge            node {$N^3$} (P1)
            edge[bend left] node {$N^3$} (\is2)
     (\is2) edge            node {$N^2$} (P2)
            edge[bend left] node {$N^2$} (\is1);
   \end{tikzpicture}
  \end{center}    
  Assume that the molecular abundance of $R$ is of order $O(N)$ and that the molecular abundance of $E$ is of order $O(1)$.
  The small amount of enzyme molecules will be consumed fast and none will be produced. Therefore, after a while, there will be no enzyme molecules present. Each intermediate molecule will fast produce $P_1$ or $P_2$ and, after that, no other reaction can possibly take place. That is, after a time of order $O(1/N)$, no reaction will take place. Thus, in order to observe the dynamics of the system, time should be rescaled by a factor $N$. That is, the time $\tilde{t}=t/N$ should be considered. This is the same as studying the SRN with all reaction rates rescaled by a factor of $1/N$.

  Despite some reaction rates tend to zero with $N$, our results can be applied to approximate the dynamics of the SRN. In particular the reduced SRN is given by
  
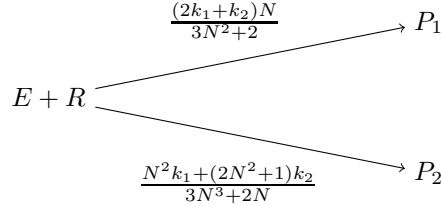
\begin{center}
   \begin{tikzpicture}
    \node[state] (E+R) at (1,0)  {$E+R$};
    \node[state] (P1)  at (6,1)  {$P_1$};
    \node[state] (P2)  at (6,-1) {$P_2$};
    \path[->]
     (E+R) edge node[pos=0.63]      {$\frac{(2k_1+k_2)N}{3N^2+2}$} (P1)
           edge node[swap,pos=0.75] {$\frac{N^2k_1+(2N^2+1)k_2}{3N^3+2N}$} (P2);
   \end{tikzpicture}
  \end{center}    
  where the magnitudes of the molecular abundances of $E$, $R$, $P_1$, $P_2$ are the same as in the full reaction network.
  \end{ex}
 
 \subsection{Assumptions}
 
 Let $(\Sp,\C,\R)$ be a SRN with a set of intermediate species $\In\subset\Sp$,  let $\K^N$ be a sequence  of kinetics indexed by $N\in\NN$, and let $X^N_\cdot$ be the corresponding stochastic process (\ref{eq:general_form_stochastic_reaction_network}). Define
 \begin{align}\label{eq:R^0}
  \R^0&\df\prg{y_i\rightarrow y_j\in\R\colon y_i\notin \In},\\
  \label{eq:R^00}
  \R^1&\df\prg{y_i\rightarrow y_j\in\R\colon y_i,y_j\notin \In}\subset\R^0.
 \end{align}
 Specifically, $\R^0$ is the set of reactions whose reactant is not an intermediate, while $\R^1$ is the set of reactions not involving intermediates at all. 
 
 Fix a non-negative vector of scaling coefficients, $\alpha=(\alpha(S))_{S\in\Sp\setminus\In}\in\RR^{\Sp\setminus\In}_{\geq0}$, and define the rescaled process, 
 \begin{equation}\label{rescale}
 \widehat{X}_t^N=N^{-\alpha} p(X_t^N),
 \end{equation}
 where $ p\colon\RR^{\Sp}\rightarrow\RR^{\Sp\setminus\In}$ is the projection onto the non-intermediate species space and the multiplication $N^{-\alpha} p(X_t^N)$ is intended component-wise.
 The process $\widehat{X}_\cdot^N$ is the rescaled process in the sense of \cite{ball:asymptotic,kurtz:rescale,popovich:rescale} for the non-intermediate species. Since $\alpha(S)$ might differ from species to species, $\widehat{X}_\cdot^N$ is a multiscale process.

\begin{ass}\label{ass:big_with_N}
 Let  $\alpha$ be given as in \eqref{rescale}.

 \begin{enumerate}[(i)]
  \item\label{ass:independence_of_the_dynamics} (Rate functions and intermediates)  We assume that $(\Sp,\C,\R)$ equipped with $\K^N$ satisfies Assumption \ref{ass:independence_of_the_dynamics_without_N} for all $N\in\NN$.
  \item\label{ass:rescale} (Rescaling of abundances) We assume that for any non-intermediate species $S\in\Sp\setminus\In$,
   \begin{equation}\label{eq:O(1)}
    \widehat{X}_t^N(S)=\mathcal{O}(1),
   \end{equation}
   that is, the scaled abundances do not blow up before time $t$.
   To make (\ref{eq:O(1)}) precise, we require that there exists $T>0$ such that for any $S\in\Sp\setminus\In$,
   \begin{subnumcases}
    \displaystyle \forall\,\nu>0\,\,\exists\,\Upsilon_{\nu}\colon\limsup_{N\rightarrow\infty}P\pr{\sup_{[0,T]}\widehat{X}_t^N(S)>\Upsilon_{\nu}}<\nu\label{eq:hypothesis_O_1_a}\\
    \displaystyle \La\prg{t\in[0,T]\colon\lim_{N\rightarrow\infty}\widehat{X}_t^N(S)=0\text{ a.s.}}=0,\label{eq:hypothesis_O_1_b}
   \end{subnumcases}
   where $\La$ denotes the usual Lebesgue measure on $\RR$. 
  \item\label{ass:ratefct} (Convergence of rate functions)
   We assume that there exist  a set of locally Lipschitz functions $\prg{\lambda_r(\cdot)}_{r\in \R^0}$ defined on $\RR^{\Sp\setminus\In}_{\geq0}$, fulfilling
   $$x\in\RR^{\Sp\setminus\In}_{>0}\Rightarrow\lambda_r(x)>0,$$
   and a set of non-negative real numbers $\prg{\beta_r}_{r\in \R^0}$ such that, for all $r\in \R^0$,
   \begin{equation}\label{eq:limit_rates}
    N^{-\beta_r}\lambda_r^N(N^{\alpha}x)\xrightarrow[N\rightarrow\infty]{}\lambda_r(x)
   \end{equation}
   uniformly on compact sets, where the rate functions $\lambda_r^N$ are extended to the real vectors by considering the floor function of the argument.
  \item\label{ass:degradation_intermediates} (Degradation of intermediates)  Let $C^N_n$, $\tau^N_n$, $t^N_n$ and $\pi^N_{\ell j}$ be as defined after Assumption \ref{ass:independence_of_the_dynamics_without_N}. Let
   $$
    \beta^*_\ell=\max_{i\in\sr}\beta_{i\ell},\qquad\alpha^*_j=\min_{S\in y_j}\alpha(S),
   $$
   where $\beta_{i\ell}$ is as in (\ref{ass:ratefct}) for $r=y_i\to \is_\ell$. Moreover, define
   \begin{equation}\label{eq:p_ell_j}
    p_{\ell j}^\varepsilon(N)\df P\prconf{\tau^N_1>\frac{N^{\alpha^*_j}\varepsilon}{N^{\beta^*_{\ell}}\pi^N_{\ell j}}}{C^N_1(t^N_1)=\is_\ell, C^N_1(t^N_1+\tau^N_1)=y_j}.
   \end{equation}
   By definition of the continuous-time Markov chains $C^N_n(\cdot)$, for any $n$
   $$P\prconf{\tau^N_n>\frac{N^{\alpha^*_j}\varepsilon}{N^{\beta^*_{\ell}}\pi^N_{\ell j}}}{C^N_n(t^N_n)=\is_\ell, C^N_n(t^N_n+\tau^N_n)=y_j}=p_{\ell j}^\varepsilon(N).$$
   We assume that the size of $\tau^N_n$ is controlled, that is, for all $\varepsilon>0$, $\ell\in\bigcup_{i\in\sr}\iii$ and $j\in\fp$, we have
   \begin{equation}\label{eq:lifespan_going_to_zero_relaxed}
    \pi^N_{\ell j}N^{\beta^*_\ell-\alpha^*_j}p_{\ell j}^\varepsilon(N)\xrightarrow[N\rightarrow\infty]{}0.
   \end{equation}
   Sufficient conditions for \eqref{eq:lifespan_going_to_zero_relaxed} are given in Propositions \ref{prop:sufficient_conditions} and \ref{prop:mu_suff_ass}.
   

  \item\label{ass:single_scale_system} (Single scale system)
   For any non-intermediate species $S\in\Sp\setminus\In$, let 
   $$\R^1_S\df\R_S\cap\R^1\qquad\text{and}\qquad\overline{\R}_S\df\prg{r\in\R^*\setminus\R^1\,\colon\,\xi_{r}(S)\neq 0}.$$
   Moreover, for all $\ell\in\iii$ and $j$ in the set of complexes indices, let $\pi^N_{\ell j}$ be as in \eqref{eq:pi}. We  assume that 
   \begin{equation}\label{eq:probability_limit_assumption}
    \begin{cases}
     \displaystyle\exists\, \gamma_{\ell j}=\lim_{N\rightarrow\infty} \log_N \pi^N_{\ell j} &\in [-\infty,0]\\
     \\
     \displaystyle\exists\, \lim_{N\rightarrow\infty} \pi^N_{\ell j}N^{-\gamma_{\ell j}} &\text{if }\gamma_{\ell j}>-\infty.
    \end{cases}
   \end{equation}
   and
   \begin{equation}\label{eq:single_scale_limit}
    \max\pr{\prg{\beta_r}_{r\in \R^1_S}\cup \prg{\beta_{i\ell}+\gamma_{\ell j}}_{\ell\in\iii,y_i\to y_j\in\overline{\R}_S}}\leq\alpha(S),
   \end{equation}
   where $\beta_r$ with $r\in\R^0$ is as in (\ref{ass:ratefct}), and $\max\emptyset=-\infty$.
  \end{enumerate}
 \end{ass}
 
 \begin{rem}\label{rem:single_scale_system}
  `Single scale system' in Assumption \ref{ass:big_with_N}(\ref{ass:single_scale_system}) refers to the time scale of the reduced SRN, as defined in \cite{popovich:rescale}.
 \end{rem}
 
 \begin{rem}
  Time rescaling in the sense of Example \ref{ex:rescaling_of_time} might be considered. It is equivalent to a rescaling of all the rate functions by a common factor, and therefore equivalent to adding a common term to all the $\beta$'s . Thus, time rescaling is implicitly considered in our framework of model reduction. We will ignore it in the development of the theory.
 \end{rem}

 \begin{rem}\label{rem:massaction_beta}
  Assume mass-action kinetics and assume that for any reaction $r\colon y_i\rightarrow y_j\in\R^0$, the constant $k^N_r$ is of the form $N^{\eta_r}k_r$ with $k_r>0$ and $\eta_r\in\RR$. Thus, 
  $$\lambda^N_r(N^{\alpha}x)=N^{\eta_r}k_r\prod_{S\in y_i}\frac{(N^{\alpha(S)}x(S))!}{(N^{\alpha(S)}x(S)-y_i(S))!}\mathbbm{1}_{\prg{N^{\alpha(S)}x(S)\geq y_i(S)}}\,.$$
  This means that the right scaling for the rate function $\lambda^N_r$ is
  $$\beta_r=\eta_r+\sum_{S\in y_i}\alpha(S)\cdot y_i(S)\,.$$
  Indeed,
  $$N^{-\beta_r}\lambda^N_r(N^{\alpha}x)\xrightarrow[N\rightarrow\infty]{}\lambda_r(x)$$
  uniformly on compact sets, where
  $$\lambda_r(x)=k_r\prod_{\substack{S\in y_i \\ \alpha(S)=0}}\frac{x(S)!}{(x(S)-y_i(S))!}\mathbbm{1}_{\prg{x(S)\geq y_i(S)}}\prod_{\substack{S\in y_i \\ \alpha(S)>0}}x(S)^{y_i(S)}\mathbbm{1}_{\prg{x(S)>0}}\,.$$
 \end{rem}
 
 \begin{rem}\label{rem:relaxed_hypothesis}
  Theorems \ref{thm:eq_true_for_bounded_rates} and \ref{thm:weak_convergence_bounded_rates} below hold even if \eqref{eq:probability_limit_assumption} and \eqref{eq:single_scale_limit} in Assumption \ref{ass:big_with_N}(v) are replaced by the weaker conditions
  \begin{gather}
   \label{eq:limsup_limited_hypothesis}
   \exists c_{\ell j}>0 \text{ s.t. }\limsup_{N\rightarrow\infty} \pi_{\ell j}^N N^{\beta_{\ell}^*-\alpha^*_j}\leq c_{\ell j}\\
   \label{eq:single_scale}
   \max\pr{\prg{\beta_r}_{r\in \R^1_S}\cup\prg{\limsup_{N\rightarrow\infty}\pr{\beta_{i\ell}+\log_N \pi_{\ell j}^N}}_{\ell\in\iii, y_i\to y_j\in\overline{\R}_S}}\leq\alpha(S).
  \end{gather}
  We will use these in the proof of Theorems \ref{thm:eq_true_for_bounded_rates} and \ref{thm:weak_convergence_bounded_rates}.
 \end{rem}

 Under the assumption that $\widehat{X}^N_0$ is bounded uniformely on $N$, condition \eqref{eq:hypothesis_O_1_a} is fulfilled for a special class of reaction networks called \emph{conservative} reaction networks (cf. Remark \ref{rem:conservative_case}).
 In order to state suffiecient conditions for \eqref{eq:lifespan_going_to_zero_relaxed} to hold, for any $\ell\in\iii$ we define
   $$a_\ell=\min_{y_j\in\FP_\ell}\alpha^*_j,$$
   where $\FP_\ell\subseteq\FP$ denotes the set of final products which are obtainable from $\is_\ell$ through a path of intermediates. In other words, $\FP_\ell$ is the set of final products $y_j$ such that there exists a path of the form
   $$\is_\ell\to\is_{\ell_1}\to\dots\to\is_{\ell_k}\to y_j.$$
   The following holds:
   \begin{prop}\label{prop:sufficient_conditions}
    Equation \eqref{eq:lifespan_going_to_zero_relaxed} holds if for all $\ell\in\bigcup_{i\in\sr}\iii$ and $\varepsilon>0$, we have
   \begin{equation}\label{eq:lifespan_going_to_zero}
    N^{\beta^*_\ell-a_\ell}P\prconf{\tau^N_1>N^{a_\ell-\beta^*_{\ell}}\varepsilon}{C^N_1(t^N_1)=\is_\ell}\xrightarrow[N\rightarrow\infty]{}0.
   \end{equation}
   Moreover, \eqref{eq:lifespan_going_to_zero_relaxed} holds if for all $\ell\in\bigcup_{i\in\sr}\iii$ and $\varepsilon>0$, we have \eqref{eq:limsup_limited_hypothesis} and
   \begin{equation}\label{eq:exp_lifespan_going_to_zero}
    N^{\beta^*_{\ell}-a_\ell}\Econf{\tau^N_1}{C^N_1(t^N_1)=\is_{\ell}}\xrightarrow[N\rightarrow\infty]{}0.
   \end{equation}
   \end{prop}
   \begin{proof}
    The first part of the proposition is proven by
    \begin{align*}
     \sum_{j\in\fp}\pi^N_{\ell j}N^{\beta^*_\ell-\alpha^*_j}p_{\ell j}^\varepsilon(N)&\leq N^{\beta^*_\ell-a_\ell}\sum_{j\in\fp}\pi^N_{\ell j}P\prconf{\tau^N_1>N^{a_\ell-\beta^*_{\ell}}\varepsilon}{C^N_1(t^N_1)=\is_\ell, C^N_1(t^N_1+\tau^N_1)=y_j}\\
     &=N^{\beta^*_\ell-a_\ell}P\prconf{\tau^N_1>N^{a_\ell-\beta^*_{\ell}}\varepsilon}{C^N_1(t^N_1)=\is_\ell}.
    \end{align*}
    The second part of the proposition follows from 
    \begin{align*}
     \sum_{j\in\fp}\pi^N_{\ell j}N^{\beta^*_\ell-\alpha^*_j}\Econf{\tau^N_1}{C^N_1(t^N_1)=\is_\ell, C^N_1(t^N_1+\tau^N_1)=y_j}&\\
     &\hspace*{-63pt}\leq N^{\beta^*_\ell-a_\ell}\sum_{j\in\fp}\pi^N_{\ell j}\Econf{\tau^N_1}{C^N_1(t^N_1)=\is_\ell, C^N_1(t^N_1+\tau^N_1)=y_j}\\
     &\hspace*{-63pt}=N^{\beta^*_\ell-a_\ell}\Econf{\tau^N_1}{C^N_1(t^N_1)=\is_{\ell}}.
    \end{align*}
    Therefore, \eqref{eq:exp_lifespan_going_to_zero} implies that for any $j\in\fp$
    $$\pi^N_{\ell j}N^{\beta^*_\ell-\alpha^*_j}\Econf{\tau^N_1}{C^N_1(t^N_1)=\is_\ell, C^N_1(t^N_1+\tau^N_1)=y_j}\xrightarrow[N\rightarrow\infty]{}0.$$
    By Markov inequality, this implies that $p_{\ell j}^\varepsilon(N)$ tends to zero as $N$ goes to infinity. By \eqref{eq:limsup_limited_hypothesis}, the latter leads to \eqref{eq:lifespan_going_to_zero_relaxed}, and the proof is complete.
   \end{proof}

   Since $\tau^N_n$ is a phase-type distributed random variable, we can express \eqref{eq:lifespan_going_to_zero} in terms of the exponential of the transition rate matrix \eqref{eq:transition_matrix}. Specifically, \eqref{eq:lifespan_going_to_zero} is equivalent to
   $$N^{\beta^*_\ell-a_\ell}(e_{\ell})^\top \exp\pr{N^{a_\ell-\beta^*_{\ell}}\varepsilon Q_{\ii,\ii}^N}e\xrightarrow[N\rightarrow\infty]{}0\qquad\forall \ell\in\bigcup_{i\in\sr}\iii,$$
   where $(e_{\ell})^\top$ denotes the transpose of the canonical base vector with a one in the $\ell$-th entry and $e$ is the vector with all entries equal to one.  
  A sufficient condition for \eqref{eq:exp_lifespan_going_to_zero} to hold is given in the proposition below:
 
 \begin{prop}\label{prop:mu_suff_ass}
  Assume Assumptions \ref{ass:big_with_N}(\ref{ass:independence_of_the_dynamics},\ref{ass:ratefct}) are fulfilled for some $\alpha\in\RR^{\Sp\setminus\In}$. For all $i\in\sr$ and $\ell\in\ii$, let $\mu_{i\ell}^N(x)$ be as in (\ref{eq:definition_of_mu}) and define
  $$\alpha^*=\min_{j\in\fp}\alpha^*_j.$$
  We have that, if
  \begin{equation}\label{eq:mu_tend_to_zero}
   N^{-\alpha^*}\mu_{i\ell}^N(N^{\alpha}x)\xrightarrow[N\rightarrow\infty]{} 0
  \end{equation}
  for all $x\in\RR^{\Sp\setminus\In}_{\ge 0}$ and for all $i\in\sr$, $\ell\in\ii$, then \eqref{eq:exp_lifespan_going_to_zero} in Assumption \ref{ass:big_with_N}(\ref{ass:degradation_intermediates}) holds. Moreover, if $\alpha^*_j=\alpha^*_{j^\prime}$ for all $j,j^\prime\in\fp$, then \eqref{eq:mu_tend_to_zero} is also a necessary condition for \eqref{eq:exp_lifespan_going_to_zero} to hold.
 \end{prop}

 We prove Proposition \ref{prop:mu_suff_ass} in Section \ref{sec:proofs}. The condition \eqref{eq:exp_lifespan_going_to_zero} is sufficient for \eqref{eq:lifespan_going_to_zero_relaxed} to hold, but it is not necessary, as shown in Example \ref{ex:negative_powers_N}. Before moving on, we make a number of remarks.
 

 \subsection{The Process $Z^N_\cdot$}
 
 In order to show that the reduced SRN provides a good approximation, under the given assumptions, of different features of the original SRN, we define a sequence of processes $Z_\cdot^N$ ad hoc. We choose them such that for any fixed $t$ the (rescaled) difference $\abs{X^N_t-Z^N_t}$ tends to zero in probability, and such that the process $Z_\cdot^N$ is distributed as the process associated with the reduced SRN. We will prove other convergence statements in Theorems \ref{thm:eq_true_for_bounded_rates}, \ref{thm:weak_convergence_bounded_rates} and \ref{thm:main_theorem}.
 
 Recall that $ p\colon\RR^{\Sp}\rightarrow\RR^{\Sp\setminus\In}$ is the projection onto the non-intermediate species space. By Assumption \ref{ass:big_with_N}(\ref{ass:independence_of_the_dynamics}), 
 the reaction rates $\lambda^N_r(\cdot)$ with $r\in\R^0$ do not depend on the counts of intermediates. That is, for any $x\in\NN^\Sp$,
 $$\lambda^N_r(x)=\bar{\lambda}^N_r(p(x)),$$
 for some function $\bar{\lambda}^N_r\colon\NN^{\Sp\setminus\In}\rightarrow\RR_{\geq0}$. For the sake of convenience, we will abuse notation and let $\bar{\lambda}^N_r(x)=\lambda^N_r(x)$ for all $x\in\RR^{\Sp\setminus\In}$.

 Given the $n$-th chain of intermediates $C^N_n(\cdot)$ appearing in relation to the process $X^N_\cdot$, we denote by $\{C^N_n(\cdot)\in C_{i\ell j}\}$ the event that $C^N_n(\cdot)$ originates from the reaction $y_i\rightarrow \is_\ell$ and eventually produces the final complex $y_j$. Such an event is measurable with respect to the $\sigma$-algebra $\sigma^N_\infty$ as introduced in \eqref{eq:sigmat}.  Furthermore, let $M^N_{i\ell j}(t)$ denote the number of the chains originated before time $t$ and such that $\{C^N_n(\cdot)\in C_{i\ell j}\}$:
 $$M^N_{i\ell j}(t)=\#\prg{n\colon C^N_n(\cdot)\in C_{i\ell j}\,,\,t^N_n\leq t}=\sum_{n=1}^{Y_{i\ell}\pr{\int_0^t \lambda^N_{i\ell}(X^N_s)ds}}\mathbbm{1}_{\{C^N_n(\cdot)\in C_{i\ell j}\}}.$$
 The processes $M^N_{i\ell j}(\cdot)$ are therefore arrival processes, and we might represent them in terms of independent and identically distributed unit-rate Poisson processes $Y_{i\ell j}(\cdot)$ such that
 \begin{equation}\label{eq:def_Y_ilj}
  M^N_{i\ell j}(t)=Y_{i\ell j}\pr{\int_0^t \pi^N_{\ell j}\lambda^N_{i\ell}(X^N_s)ds}.
 \end{equation}
 In this context, $Y_{i\ell}(t)=\sum_{j\in\fp}Y_{i\ell j}(t)$. Moreover, let $t_{i\ell j,n}^N$ be the time of the $n$-th jump of the process $M^N_{i\ell j}(\cdot)$, and let $\tau_{i\ell j,n}^N$ be a collection of independent random variables distributed as $\tau^N_1$ given $(C^N_1\in C_{i \ell j})$. We now consider the process counting the number of chains of intermediates $C^N_n(\cdot)$ consumed before time $t$ and such that $\{C^N_n(\cdot)\in C_{i\ell j}\}$. Such a process is distributed as
 $$\overline{M}_{i\ell j}^N(t)=\sum_{n=1}^{M^N_{i\ell j}(t)}\mathbbm{1}_{\{t_{i\ell j,n}^N+\tau_{i\ell j,n}^N\leq t\}}.$$
 For any time $t$, we have $M_{i\ell j}^N(t)\geq \overline{M}_{i\ell j}^N(t)$. The process $\widehat{X}^N_\cdot$ can be equivalently expressed as
 \begin{equation}\label{eq:X_t^N_hat}
  \widehat{X}^N_t=\widehat{X}^N_0+N^{-\alpha}\prq{\sum_{r\in\R^1}\xi_r Y_r\pr{\int_0^t \lambda^N_r(X^N_s)ds}+\sum_{i\in\sr}\sum_{j\in\fp}\pr{y_j\sum_{\ell\in\iii}\overline{M}_{i\ell j}^N(t)-y_i\sum_{\ell\in\iii} M_{i\ell j}^N(t)}},
 \end{equation}
 where the Poisson processes $Y_r(\cdot)$ are the same as those appearing in \eqref{eq:general_form_stochastic_reaction_network}. We will use this representation in the remaining part of the paper.
 
 We define the process $\widehat{Z}^N_\cdot$ on $N^{-\alpha}\NN^{\Sp\setminus\In}$ as 
 \begin{equation}\label{eq_Z}
  \widehat{Z}^N_t=\widehat{Z}^N_0+N^{-\alpha}\prq{\sum_{r\in\R^1}\xi_r Y_r\pr{\int_0^t \lambda^N_r(Z^N_s)ds}+\sum_{i\in\sr}\sum_{j\in\fp}(y_j-y_i)\sum_{\ell\in\iii} Y_{i\ell j}\pr{\int_0^t \pi^N_{\ell j}\lambda^N_{i\ell}(Z^N_s)ds}}.
 \end{equation}
 For any fixed $t\geq 0$, the random variables $\widehat{X}^N_t$ and $\widehat{Z}^N_t$ are measurable with respect to
 $$\sigma\pr{Y_r(s),Y_{i\ell j}(s),\tau_{i\ell j,n}^N\colon r\in\R^1, i\in\sr,\ell\in\iii,j\in\fp, n,N\in\NN \text{ and }0\leq s<\infty}.$$
 The above $\sigma$-algebra contains information about the Poisson processes $Y_r(\cdot)$ for reactions not involving intermediates, about the Poisson processes  $Y_{i\ell j}(\cdot)$ that drive $M^N_{i\ell j}(\cdot)$ and about the \emph{delays} $\tau_{i\ell j,n}^N$ of the reactions proceeding through intermediates species. It does not contain full information on the intermediate chains $C^N_n(\cdot)$, but that is not required in the description of the processes $\widehat{X}^N_\cdot$ and $\widehat{Z}^N_\cdot$. The   random variables we are interested in will all be measurable with respect to the above $\sigma$-algebra, and therefore are defined on the same probability space. Since $\widehat{Z}^N_t$ is, up to rescaling, expressed in the form (\ref{eq:general_form_stochastic_reaction_network}), it is distributed as the rescaled stochastic process associated with \eqref{eq:reduced_network}.

 There is a precise intuition behind the choice of $\widehat{Z}^N_t$ as approximating process for the original system. Consider \eqref{eq:X_t^N_hat}: if (\ref{eq:lifespan_going_to_zero_relaxed}) holds, then we expect the lifetime of the intermediate species to decrease with $N$. Thus, we could imagine that, for any fixed time $t$, $M_{i\ell j}^N(t)=\overline{M}_{i\ell j}^N(t)$ with high probability and, thus, that $\widehat{X}^N_t$ is approximated by
 \begin{equation}\label{eq:W^N_t}
  \widehat{W}^N_t\df\widehat{X}^N_0+N^{-\alpha}\pr{\sum_{r\in\R^1}\xi_r Y_r\pr{\int_0^t \lambda^N_r(X^N_s)ds}+\sum_{i\in\sr}\sum_{j\in\fp}(y_j-y_i)\sum_{\ell\in\iii} M_{i\ell j}^N(t)}.
 \end{equation}
 The process $\widehat{Z}_\cdot^N$ in (\ref{eq_Z}) is defined analogously to (\ref{eq:W^N_t}).
 
 Unfortunately, we cannot hope for $\widehat{X}_\cdot^N$ to converge weakly to $\widehat{Z}_\cdot^N$ in the Skorohod topology in general (cf.\ Example \ref{ex:weak_convergence}). However, we will show a uniform punctual convergence in probability as well as convergence in occupation measure for the difference of the stopped processes $\widehat{X}_{\cdot\wedge T}^N$ and $\widehat{Z}_{\cdot\wedge T}^N$, for any fixed $T>0$. Furthermore, we will give additional hypothesis under which the convergence in probability in the Skorohod space holds.
 
 \subsection{Bounded Reaction Rates}
 
 Recall that $\R^0$ in (\ref{eq:R^0}) is the set of reactions whose reactant is not an intermediate. Here we are concerned with the case when all reaction rates of reactions in $\R^0$ are bounded by a power of $N$, specifically for any $r\in\R^0$,
  \begin{equation}\label{eq:bounded_rates}
   N^{-\beta_r}\lambda^N_r(x)\leq B_r\quad\forall\,N\in\NN,\,\forall\,x\in \RR^{\Sp\setminus\In}_{\ge 0},
  \end{equation}
  where $\beta_r$ is as in Assumption \ref{ass:big_with_N}(\ref{ass:ratefct}) and $B_r$ is a positive constant (later the constant will also be referred to as $B_{i \ell}$ if in relation to the reaction $y_i\to H_\ell$). It is worth mentioning that in this case, \eqref{eq:hypothesis_O_1_a} in Assumption \ref{ass:big_with_N}(\ref{ass:rescale}) is always fulfilled if $\widehat{X}^N_0$ is stochastically bounded (cf.\ Remark \ref{rem:bounded_rates_O_1}). This is desirable because it suffices to control stochastic boundeness of a real random variable rather than of an entire stochastic process. Moreover, \eqref{eq:bounded_rates} can be assume to hold if the network is conservative and $\widehat{X}^N_0$ is bounded independently of $N$ (cf. Remark \ref{rem:conservative_case}).
 
 The proofs of Theorems \ref{thm:eq_true_for_bounded_rates} and \ref{thm:weak_convergence_bounded_rates} can be found in Section \ref{sec:proofs_bounded_rates}, using the relaxed version of Assumption \ref{ass:big_with_N}(\ref{ass:single_scale_system}) as given in Remark \ref{rem:relaxed_hypothesis}. The weaker condition is sufficient to prove Corollary \ref{cor:cor_to_bounded_rate_lemma} as well. 
 
 \begin{thm}\label{thm:eq_true_for_bounded_rates}
  Assume Assumption \ref{ass:big_with_N} is fulfilled for some $\alpha\in\RR^{\Sp\setminus\In}$. Further, assume that
  $$E\prq{\abs{\widehat{X}^N_0-\widehat{Z}^N_0}}\xrightarrow[N\rightarrow\infty]{}0,$$
  and that the initial amounts of the intermediate species are 0. Finally, assume that for any $r\in\R^0$, \eqref{eq:bounded_rates} holds and $\lambda_r$ is Lipschitz. Then, if $T$ is as in Assumption \ref{ass:big_with_N}(\ref{ass:rescale}), we have that
  \begin{equation}\label{eq_expectation_goes_to_zero}
   \sup_{t\in[0,T]}E\prq{\abs{\widehat{X}^N_t-\widehat{Z}^N_t}}\xrightarrow[N\rightarrow\infty]{}0,
  \end{equation}
  In particular, (\ref{eq_expectation_goes_to_zero}) implies  that for  all $\varepsilon>0$, 
  \begin{equation}\label{eq_to_be_proved}
   \sup_{t\in[0,T]}P\pr{\abs{\widehat{X}^N_t-\widehat{Z}^N_t}>\varepsilon}\xrightarrow[N\rightarrow\infty]{}0.
  \end{equation}
  Finally, for any continuous function $f\colon\RR^{\Sp\setminus\In}\to\RR$ we have
  \begin{equation}\label{eq:occupation_measure_bounded_rates}
   P\pr{\sup_{t\in[0,T]}\abs{\int_0^t\pr{f(\widehat{X}^N_s)-f(\widehat{Z}^N_s)}ds}>\varepsilon}\xrightarrow[N\rightarrow\infty]{}0.
  \end{equation}
 \end{thm}

 \begin{rem}\label{rem:bounded_rates_O_1}
  Assume that \eqref{eq:bounded_rates}, \eqref{eq:limsup_limited_hypothesis}, and \eqref{eq:single_scale} hold. Assume further that $\widehat{X}^N_0$ is stochastically bounded, meaning that for every $\nu>0$ there exists $\Upsilon_\nu$ such that for every $S\in\Sp\setminus\In$
  $$\limsup_{N\rightarrow\infty}P\pr{\widehat{X}^N_0(S)>\Upsilon_\nu}<\nu.$$
  Our aim is to prove \eqref{eq:hypothesis_O_1_a}. By \eqref{eq:X_t^N_hat}
  \begin{multline*}
   \sup_{t \in[0,T]}\widehat{X}_t^N(S)\leq X^N_0(S)+N^{-\alpha(S)}\sum_{r\in\R_S^1}\abs{\xi_r(S)}Y_r(N^{\beta_r}B_r T)+\\
   +N^{-\alpha(S)}\sum_{i\in\sr}\sum_{j\in\fp}2\pr{y_j(S)+y_i(S)}\sum_{\ell\in\iii}Y_{i\ell j}(\pi^N_{\ell j}N^{\beta_{i\ell}}B_{i\ell}T),
  \end{multline*}
  where $\R^1_S$ is defined according to (\ref{eq:R_S}). Using assumptions \eqref{eq:limsup_limited_hypothesis}, \eqref{eq:single_scale} and the Law of Large Numbers for Poisson processes to control the above expression for $\alpha(S)>0$, we obtain that, for any $\nu>0$, there exists $\Upsilon_\nu^\prime>0$, such that
  $$\limsup_{N\rightarrow\infty}P\pr{\sup_{t\in[0,T]}\widehat{X}_t^N(S)>\Upsilon_\nu^\prime}<\nu.$$
 \end{rem}
 
 \begin{rem}\label{rem:conservative_case}
  Conservative reaction networks are a special class of reaction networks \citep{horn:general_mass_action}. In a conservative reaction network, a positive linear combination of the species abundances is preserved throughout time and, hence, the total abundances are bounded from above given any initial condition.
  In such class of reaction networks, if $\widehat{X}^N_0$ is bounded uniformly on $N$ then condition (\ref{eq:bounded_rates}) is fulfilled. Indeed, if the original reaction network is conservative, then the reduced reaction network is conservative as well \citep{feliu:intermediates}. Let $\mathcal{S}_1$ and $\mathcal{S}_2$ denote the spaces spanned by the reaction vectors of the original and of the reduced network, respectively. Moreover let $\mathcal{S}=p(\mathcal{S}_1)\cup \mathcal{S}_2\subset \RR^{\Sp\setminus\In}$. It can be shown that $\mathcal{S}_2\subseteq p(\mathcal{S}_1)$, but this lies outside our concerns. The initial condition $\widehat{X}^N_0$ varies in a compact set $K_0$. Therefore, for any $r\in\R^0$, we might consider a modified version of the rate functions $\lambda_r^N$,  such that
  $$N^{-\beta_r}\lambda_r^N(N^{\alpha}x)=1\quad\forall x\notin (S+K_1)\cap\RR^{\Sp\setminus\In}_{\geq0},$$
  and $K_1\supset K_0$ is a compact set. Thus, the limit functions $\lambda_r$  in Assumption \ref{ass:big_with_N}(\ref{ass:ratefct})
  are 1 outside a compact set and therefore bounded. Due to \eqref{eq:limit_rates}, condition (\ref{eq:bounded_rates}) is met.
  In particular, it follows from Remark \ref{rem:bounded_rates_O_1} that in this case \eqref{eq:hypothesis_O_1_a} always holds.
 \end{rem}
 
 \begin{cor}\label{cor:cor_to_bounded_rate_lemma}
  Assume that the assumptions of Theorem \ref{thm:eq_true_for_bounded_rates} hold. Then, the difference between the processes $\widehat{X}^N_{\cdot\wedge T}$ and $\widehat{Z}^N_{\cdot\wedge T}$ converges in finite dimensional distribution to 0.
 \end{cor}
 \begin{proof}
  From Theorem \ref{thm:eq_true_for_bounded_rates} we have that (\ref{eq_to_be_proved}) holds for any $\varepsilon>0$. Thus, for any finite set of time points $\prg{t_m}_{m=0}^p\subseteq[0,T]$ we have that
  \begin{align*}
   P\pr{\max_{0\leq m\leq p}\abs{\widehat{X}^N_{t_m}-\widehat{Z}^N_{t_m}}>\varepsilon}&=P\pr{\bigcup_{m=0}^p\prg{\abs{\widehat{X}^N_{t_m}-\widehat{Z}^N_{t_m}}>\varepsilon}}\\
   &\leq\sum_{m=0}^p P\pr{\abs{\widehat{X}^N_{t_m}-\widehat{Z}^N_{t_m}}>\varepsilon}\xrightarrow[N\rightarrow\infty]{}0,
  \end{align*}
  hence the corollary holds.
 \end{proof}

 We discuss here some applications of Theorem \ref{thm:eq_true_for_bounded_rates} and Corollary \ref{cor:cor_to_bounded_rate_lemma}.
  
 \begin{ex}\label{ex:main_example_works}
  Consider the reaction network in Example \ref{ex:main_example}. Assumption \ref{ass:big_with_N}(\ref{ass:independence_of_the_dynamics}) holds. Further, if we let $\alpha(E)=0$ and $0<\alpha(R)=\alpha(P_1)=\alpha(P_2)<2$, then Assumption \ref{ass:big_with_N}(\ref{ass:rescale}-\ref{ass:single_scale_system}) are satisfied if we choose the initial value  $X_0^N$ proportional to the scaling $N^\alpha$ and $\beta_r$ according to Remark \ref{rem:massaction_beta}. Note that the reaction network is conservative in the sense of Remark \ref{rem:conservative_case}. Thus, \eqref{eq:bounded_rates} holds and by Theorem \ref{thm:eq_true_for_bounded_rates} and Corollary \ref{cor:cor_to_bounded_rate_lemma}, the probability distribution of the process associated with the reduced SRN approximates, in the sense of Theorem \ref{thm:eq_true_for_bounded_rates} and Corollary \ref{cor:cor_to_bounded_rate_lemma}, the probability distribution of the process \eqref{eq:example_with_rates_small}.
 \end{ex}
  
 \begin{ex}
  Consider the Michaelis-Menten mechanism taken with mass-action kinetics:
  
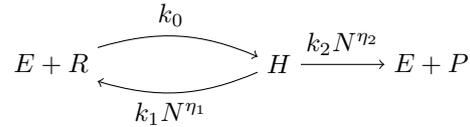
\begin{center}
   \begin{tikzpicture}
    \node[state] (E+R) at (0,0) {$E+R$};
    \node[state] (\is) at (3,0) {$\is$};
    \node[state] (E+P) at (5,0) {$E+P$};
    \path[->] 
     (E+R) edge[bend left] node[pos=.45]{$k_0$} (\is)
     (\is) edge[bend left] node[pos=.55]{$k_1N^{\eta_1}$} (E+R)
           edge            node{$k_2N^{\eta_2}$} (E+P);
   \end{tikzpicture}
  \end{center}
  Assumption \ref{ass:big_with_N}(\ref{ass:independence_of_the_dynamics}) is satisfied, as well as (\ref{eq:bounded_rates}) since the network is conservative. The probability that a molecule of $\is$ is transformed into the complex $E+R$ is $k_1N^{\eta_1}/\pr{k_1N^{\eta_1}+k_2N^{\eta_2}}$, while the probability that it is transformed into the complex $E+P$ is $k_2N^{\eta_2}/\pr{k_1N^{\eta_1}+k_2N^{\eta_2}}$. The reduced SRN is given by
  
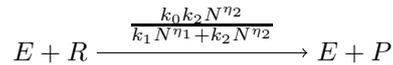
\begin{center}
   \begin{tikzpicture}
    \node[state] (E+R) at (1,0) {$E+R$};
    \node[state] (E+P) at (5,0) {$E+P$};
    \path[->] 
     (E+R) edge node{$\frac{k_0k_2N^{\eta_2}}{k_1N^{\eta_1}+k_2N^{\eta_2}}$} (E+P);
   \end{tikzpicture}
  \end{center}
  If we let that $\alpha(E)=0$, $\alpha(R)<\eta_1\vee\eta_2$ and $\alpha(P)=\alpha(R)\wedge\pr{\alpha(R)+\eta_2-\eta_1}$, then Assumption \ref{ass:big_with_N}(\ref{ass:rescale}-\ref{ass:single_scale_system}) are satisfied if we choose the initial value  $X_0^N$ proportional to the scaling $N^\alpha$ and $\beta_r$ according to Remark \ref{rem:massaction_beta}. In this case, Theorem \ref{thm:eq_true_for_bounded_rates} and Corollary \ref{cor:cor_to_bounded_rate_lemma} state in which sense the original process is approximated by the one associated with the reduced SRN. The magnitudes of the molecular abundances are the same as in the original system.
  
  In the reduced SRN the amount of enzyme $E$ is conserved. Hence, the model can further be reduced to
  
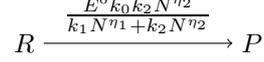
\begin{center}
   \begin{tikzpicture}
    \node[state] (R) at (1,0) {$R$};
    \node[state] (P) at (4,0) {$P$};
    \path[->] 
     (R) edge node{$\frac{E^0k_0k_2N^{\eta_2}}{k_1N^{\eta_1}+k_2N^{\eta_2}}$} (P);
   \end{tikzpicture}
  \end{center}
  where the amount of $E$ molecules  constantly equals $E^0$.
  
  Let $\delta=\alpha(R)+\min\prg{0,\eta_2-\eta_1}$. If $\delta<0$, we wait a time of order $O(N^{-\delta})$ for the first reaction to occur in the reduced SRN. Thus, we might rescale time in the original SRN by $\widetilde{t}=N^{\delta}t$. As shown in example \ref{ex:rescaling_of_time}, this is equivalent to rescale the rate functions. After rescaling, reduction can be performed again to obtain an approximation of the system's dynamics.
 \end{ex}
 
 The following example concerns a network where not all the rates out of intermediate states are high. Moreover, it shows that condition \eqref{eq:exp_lifespan_going_to_zero} is sufficient for \eqref{eq:lifespan_going_to_zero_relaxed} in Assumption \ref{ass:big_with_N}(\ref{ass:degradation_intermediates}) to hold, but it is not necessary.
 
 \begin{ex}\label{ex:negative_powers_N}
 Consider the SRN taken with mass-action kinetics,
  
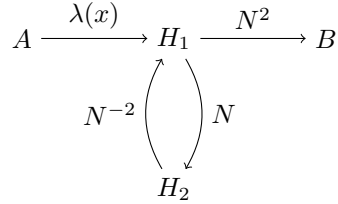
\begin{center}
   \begin{tikzpicture}
    \node[state] (A)    at (1,3) {$A$};
    \node[state] (\is1) at (3,3) {$\is_1$};
    \node[state] (\is2) at (3,1) {$\is_2$};
    \node[state] (B)    at (5,3) {$B$};
    \path[->] 
     (A)  edge node{$\lambda(x)$} (\is1)
     (\is1) edge node{$N^2$}        (B)   
          edge[bend left=30] node{$N$}      (\is2)   
     (\is2) edge[bend left=30] node{$N^{-2}$} (\is1);   
   \end{tikzpicture}
  \end{center}
  with $\alpha(A)=\alpha(B)=0$. Assumption \ref{ass:big_with_N} is fulfilled if we choose the initial value  $X_0^N$ proportional to the scaling $N^\alpha$ and $\beta_r$ according to Remark \ref{rem:massaction_beta}. This is true even though the consumption rate of $\is_2$ tends to zero. Moreover, the reaction network is conservative, thus by Theorem \ref{thm:eq_true_for_bounded_rates}, the reduced SRN
  
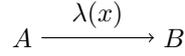
\begin{center}
   \begin{tikzpicture}
    \node[state] (A)  at (1,3) {$A$};
    \node[state] (B)  at (3,3) {$B$};
    \path[->] 
     (A)  edge node{$\lambda(x)$} (B);
   \end{tikzpicture}
  \end{center}
  provides a good approximation of the dynamics of the original SRN, for $N$ large.
  
  Further, \eqref{eq:lifespan_going_to_zero} holds since for any fixed $\varepsilon>0$, the probability that a chain of intermediates survives for a time bigger than $\varepsilon$ goes to zero with $N\rightarrow\infty$. Hence by Proposition \ref{prop:sufficient_conditions}
  \eqref{eq:lifespan_going_to_zero_relaxed} holds as well.  However, in this case \eqref{eq:exp_lifespan_going_to_zero} does not hold. If we denote $A=y_3$ and $B=y_4$, this can be shown by making use of Proposition \ref{prop:mu_suff_ass} and 
  $$\mu_{32}^N(x)=\frac{Nk\lambda(x)}{N^2\cdot N^{-2}}=Nk\lambda(x)\xrightarrow[N\rightarrow\infty]{}\infty\quad\text{for any }x\in\RR^{\Sp\setminus\In}_{\ge 0}.$$ 
 \end{ex}
  
 For the particular case $\alpha>0$, a stronger convergence result than those stated in Theorem \ref{thm:eq_true_for_bounded_rates} holds. The result does not hold generally for all $\alpha$, as shown in Example \ref{ex:weak_convergence}.
  
 \begin{thm}\label{thm:weak_convergence_bounded_rates}
  Assume the assumptions of Theorem \ref{thm:eq_true_for_bounded_rates} are fulfilled and  that $\alpha>0$. Then, for any $\varepsilon>0$,
  \begin{equation}\label{eq:weak_convergence_bounded_rates}
   P\pr{\sup_{t\in[0,T]}\abs{\widehat{X}^N_{t}-\widehat{Z}^N_{t}}>\varepsilon}\xrightarrow[N\rightarrow\infty]{}0.
  \end{equation}
  In particular, this implies that the difference between the processes $\widehat{X}^N_{\cdot\wedge T}$ and $\widehat{Z}^N_{\cdot\wedge T}$ converges weakly to 0 in the Skorohod topology.
 \end{thm}
   
 \subsection{Unbounded Reaction Rates}
 In this section, we will relax the hypothesis of boundedness in Theorem \ref{thm:eq_true_for_bounded_rates}. To begin with, we introduce some new notation. Assume Assumption \ref{ass:big_with_N} is fulfilled and let $\R^*$ be defined as in \eqref{eq:reduced_network}. Define
 $$\beta^*_{ij}=\max_{\ell\in\iii}\prg{\beta_{ij},\beta_{i\ell}+\gamma_{\ell j}},$$
 where $\beta_{ij},\beta_{i\ell}$ is as in Assumption \ref{ass:big_with_N}(\ref{ass:ratefct}).  We have that for any reaction $r\in \R^*$,
 \begin{equation}\label{eq:def_lambda^*}
  N^{-\beta^*_r}\lambda_r^{N,*}(Z^N_t)\xrightarrow[N\rightarrow\infty]{}\lambda^*_r(\widehat{Z}_t),
 \end{equation}
 where $\lambda_r^{N,*}(\cdot)$ is defined in \eqref{eq:def_of_rates_Z} and $\prg{\lambda^*_r}_{r\in \R^*}(\cdot)$ is a set of locally Lipschitz functions such that
 $$v\in\RR^{\Sp}_{>0}\Rightarrow\lambda^*_r\pr{v}>0$$
 (Assumption \ref{ass:big_with_N}(\ref{ass:ratefct})).
 As in \cite{popovich:rescale}, we distinguish between fast and slow reactions. Let
 \begin{align*}
  \R^f&=\bigcup_{S \colon \alpha(S)>0} \prg{y_i\rightarrow y_j\in\R^*_S \colon \alpha(S)=\beta^*_{ij}}\\
  \R^s&=\bigcup_{S \colon \alpha(S)=0} \prg{y_i\rightarrow y_j\in\R^*_S \colon \alpha(S)=\beta^*_{ij}}.
 \end{align*}
 Moreover, let the vector $\xi^*_r\in\RR^{\Sp}$ be defined by its entries
 $$\xi^*_r(S)\df\lim_{N\rightarrow\infty}N^{\beta^*_r-\alpha(S)}\xi_r(S).$$
 Specifically, $\xi^*_r(S)=\xi_r(S)$, if $\alpha(S)=\beta^*_r$, and $\xi^*_r(S)=0$, otherwise.
 
 \begin{lem}\label{lemma:limit_process}
  
  Assume Assumption \ref{ass:big_with_N} is fulfilled for some $\alpha\in\RR^{\Sp\setminus\In}$ and let $T$ be as in Assumption \ref{ass:big_with_N}(\ref{ass:independence_of_the_dynamics}). Assume that up to time $T$, there exists a unique and almost surely well-defined solution to the equation
  \begin{equation}\label{eq:limit_process}
    Z^*_t=Z^*_0+\sum_{r\in\R^s}\xi^*_r Y_r\pr{\int_0^t \lambda^*_r(Z^*_s)ds}+\sum_{r\in\R^f}\xi^*_r\int_0^t\lambda^*_r(Z^*_s)ds,
  \end{equation} 
  where the functions $\lambda^*_r$ are the limit functions (\ref{eq:def_lambda^*}). Then, if $\widehat{Z}^N_0$ converges in probability to $Z^*_0$, the process $\widehat{Z}^N_{\cdot\wedge T}$ converges in probability to $Z^*_{\cdot\wedge T}$ with respect to the Skorohod distance.
 \end{lem}
 
 \begin{proof}
  Just note that, in our setting, $Z^N_\cdot$ is the process associated to a \emph{single-scale system} satisfying the condition of Lemma 2.8 in \cite{popovich:rescale}, and the result follows.
 \end{proof}

 \begin{ex}
  Consider again Example \ref{ex:main_example}. In Example \ref{ex:main_example_works}, we saw that the reduced SRN approximate the behaviour of  \eqref{eq:example_with_rates_small} for $N$ large, in the sense of Theorem \ref{thm:eq_true_for_bounded_rates} and Corollary \ref{cor:cor_to_bounded_rate_lemma}. Here we present a weak limit for the process of the reduced reaction network, given by Lemma \ref{lemma:limit_process}. It is easy to check that the probabilities $\pi^N_{13}$, $\pi^N_{14}$, $\pi^N_{23}$ and $\pi^N_{24}$ tend to $2/3$, $1/3$, $1/3$ and $2/3$, respectively, for $N\rightarrow\infty$. The weak limit is given by the deterministic system
  $$\begin{cases}
     \displaystyle x_t(E)=x_0(E)\\
     \displaystyle x_t(R)=x_0(R)+x_0(E)\int_0^t\prbig{k_3x_s(P_2)-(k_1+k_2)x_s(R)}ds\\
     \displaystyle x_t(P_1)=x_0(P_1)+x_0(E)\int_0^t\frac{2k_1+k_2}{3}x_s(R)ds\\
     \displaystyle x_t(P_2)=x_0(P_2)+x_0(E)\int_0^t\pr{\frac{k_1+2k_2}{3}x_s(R)-k_3x_s(P_2)}ds,
    \end{cases}$$
  where, according to the choice of $\alpha$, the counts of the species $E$ and the (scaled) concentrations of the species $R,P_1,P_2$ are considered.
 \end{ex}
 
 \begin{thm}\label{thm:main_theorem}
  Assume that the hypotheses of Lemma \ref{lemma:limit_process} are satisfied. Moreover, assume that both $\widehat{X}^N_0$ and $\widehat{Z}^N_0$ converge in probability to $Z^*_0$.
  Then, for any $\varepsilon>0$,
  \begin{equation}\label{eq_to_be_proved_unbounded_rates}
   \sup_{t\in[0,T]}P\pr{\abs{\widehat{X}^N_t-Z^*_t}>\varepsilon}\xrightarrow[N\rightarrow\infty]{}0.
  \end{equation}
  Moreover, for any continuous function $f\colon\RR^{\Sp\setminus\In}\to\RR$ we have
  \begin{equation}\label{eq:occupation_measure_unbounded_rates}
   P\pr{\sup_{t\in[0,T]}\abs{\int_0^t\pr{f(\widehat{X}^N_s)-f(Z^*_s)}ds}>\varepsilon}\xrightarrow[N\rightarrow\infty]{}0.
  \end{equation}
  Finally, if $\alpha>0$ then
  \begin{equation}\label{eq:weak_convergence}
   P\pr{\sup_{t\in[0,T]}\abs{\widehat{X}^N_{t}-Z^*_{t}}>\varepsilon}\xrightarrow[N\rightarrow\infty]{}0.
  \end{equation}
  The latter gives weak convergence of $\widehat{X}^N_{\cdot\wedge T}$ to $Z^*_{\cdot\wedge T}$ in the Skorohod topology.
 \end{thm}

 \begin{proof}
  Since $Z^*$ is almost surely unique and well defined, we have that for any $\nu>0$, there exists a constant $\Psi_\nu>0$ such that
  $$P\pr{\sup_{t\in[0,T]}\abs{Z^*_t}>\Psi_\nu}<\nu.$$
  Since the number of species is finite, due to (\ref{eq:hypothesis_O_1_a}) there exists some constant $\Upsilon^*_{\nu}>0$ such that for $N$ large enough
  $$P\pr{\sup_{t\in[0,T]}\abs{\widehat{X}_t^N}>\Upsilon^*_{\nu}}<\nu.$$
  Let
  $$\Psi_\nu^*=\max\prg{\Psi_\nu,\Upsilon^*_{\nu}}.$$
  Moreover, let $D(h)$ denote the disc of radius $h$ in $\RR_{\geq0}^{\Sp\setminus \In}$ centred in the origin, with respect to the euclidean norm. For any $r\in\R^0$, we define $\lambda^N_{b,r}(\cdot)$ such that
  $$\lambda^N_{b,r}(x)=\begin{cases}
                        \lambda^N_r(x)&\text{if }x\in D(\Psi^*_\nu)\\
                        (1-|x|-\Psi^*_\nu)\lambda^N_r\pr{\frac{\Psi^*_\nu}{|x|}x}+(|x|-\Psi^*_\nu)N^{\beta_r}&
                        \text{if }x\in D(\Psi^*_\nu+1)\setminus D(\Psi^*_\nu)\\
                        N^{\beta_r}&\text{otherwise.}   
                       \end{cases}
  $$
  These functions are Lipschitz and define a new kinetics $\K^N_b$. Let $X^N_{b,\cdot}$, $Z^N_{b,\cdot}$ and $Z^*_{b,\cdot}$ be the corresponding processes, with
  \[X^N_{b,0}=X^N_0\mathbbm{1}_{D(\Psi^*_\nu)}(\widehat{X}^N_0), \quad Z^N_{b,0}=Z^N_0\mathbbm{1}_{D(\Psi^*_\nu)}(\widehat{Z}^N_0)\quad\text{and}\quad Z^*_{b,0}=Z^*_0\mathbbm{1}_{D(\Psi^*_\nu)}(Z^*_0).\]
  With this choice, we have
  \[P\pr{X^N_{b,0}=X^N_0}\geq 1-\nu,\quad P\pr{Z^N_{b,0}=Z^N_0}\geq 1-\nu\quad\text{and}\quad P\pr{Z^*_{b,0}=Z^*_0}\geq 1-\nu,\]
  at least for $N$ large enough (by hypothesis $\widehat{Z}^N_0$ converges in probability to $Z^*_0$). Therefore
  \begin{align*}
   P\pr{\sup_{t\in[0,T]}\abs{Z^*_{b,t}}>\Psi^*_\nu}&\leq P\pr{\sup_{t\in[0,T]}\abs{Z^*_t}>\Psi^*_\nu}+\nu<2\nu,\\
   P\pr{\sup_{t\in[0,T]}\abs{\widehat{Z}^N_t}>\Psi^*_\nu}&\leq P\pr{\sup_{t\in[0,T]}\abs{\widehat{Z}^N_{b,t}}>\Psi^*_\nu}+\nu.
  \end{align*}
  The rates $\lambda^N_{b,r}(\cdot)$ satisfy the condition in Theorem \ref{thm:eq_true_for_bounded_rates} and
  $$E\prq{\abs{\widehat{X}^N_{b,0}-\widehat{Z}^N_{b,0}}}\xrightarrow[N\to\infty]{}0.$$
  From Theorem \ref{thm:eq_true_for_bounded_rates}, we have
  $$\sup_{t\in[0,T]}P\pr{\abs{\widehat{X}^N_{b,t}-\widehat{Z}^N_{b,t}}>\varepsilon}\xrightarrow[N\rightarrow\infty]{}0,$$
  and by Lemma \ref{lemma:limit_process},
  $$P\pr{\sup_{t\in[0,T]}\abs{\widehat{Z}^N_t}>\Psi^*_\nu}\leq P\pr{\sup_{t\in[0,T]}\abs{\widehat{Z}^N_{b,t}}>\Psi^*_\nu}+\nu\xrightarrow[N\rightarrow\infty]{}P\pr{\sup_{t\in[0,T]}\abs{Z^*_{b,t}}>\Psi^*_\nu}+\nu<3\nu.$$  
  Putting it all together, we have
  \begin{multline*}
    \limsup_{N\rightarrow\infty}\sup_{t\in[0,T]}P\pr{\abs{\widehat{X}^N_t-\widehat{Z}^N_t}>\varepsilon}\\
    \leq\limsup_{N\rightarrow\infty}\sup_{t\in[0,T]}P\pr{\abs{\widehat{X}^N_t-\widehat{Z}^N_t}>\varepsilon,\sup_{t\in[0,T]}\pr{\abs{\widehat{X}^N_t}\vee\abs{\widehat{Z}^N_t}}>\Psi^*_\nu}+\\
    +\limsup_{N\rightarrow\infty}\sup_{t\in[0,T]}P\pr{\abs{\widehat{X}^N_t-\widehat{Z}^N_t}>\varepsilon,\sup_{t\in[0,T]}\pr{\abs{\widehat{X}^N_t}\vee\abs{\widehat{Z}^N_t}}\leq\Psi^*_\nu}\\
    \leq\limsup_{N\rightarrow\infty}P\pr{\sup_{t\in[0,T]}\pr{\abs{\widehat{X}^N_t}\vee\abs{\widehat{Z}^N_t}}>\Psi^*_\nu}+\limsup_{N\rightarrow\infty}\sup_{t\in[0,T]}P\pr{\abs{\widehat{X}^N_{b,t}-\widehat{Z}^N_{b,t}}>\varepsilon} <4\nu.
  \end{multline*}
  Since $\nu>0$ is arbitrary, we have \eqref{eq_to_be_proved}. Similarly,
  \begin{multline*}
    \limsup_{N\rightarrow\infty}P\pr{\sup_{t\in[0,T]}\abs{\int_0^t\pr{f(\widehat{X}^N_s)-f(\widehat{Z}^N_s)}ds}>\varepsilon}\\
    \leq\limsup_{N\rightarrow\infty}P\pr{\sup_{t\in[0,T]}\abs{\int_0^t\pr{f(\widehat{X}^N_s)-f(\widehat{Z}^N_s)}ds}>\varepsilon,\sup_{t\in[0,T]}\pr{\abs{\widehat{X}^N_t}\vee\abs{\widehat{Z}^N_t}}>\Psi^*_\nu}+\\
    +\limsup_{N\rightarrow\infty}P\pr{\sup_{t\in[0,T]}\abs{\int_0^t\pr{f(\widehat{X}^N_s)-f(\widehat{Z}^N_s)}ds}>\varepsilon,\sup_{t\in[0,T]}\pr{\abs{\widehat{X}^N_t}\vee\abs{\widehat{Z}^N_t}}\leq\Psi^*_\nu}\\
    \leq\limsup_{N\rightarrow\infty}P\pr{\sup_{t\in[0,T]}\pr{\abs{\widehat{X}^N_t}\vee\abs{\widehat{Z}^N_t}}>\Psi^*_\nu}+\limsup_{N\rightarrow\infty}P\pr{\sup_{t\in[0,T]}\abs{\int_0^t\pr{f(\widehat{X}^N_{b,s})-f(\widehat{Z}^N_{b,s})}ds}>\varepsilon} <4\nu,
  \end{multline*}
  which implies that \eqref{eq:occupation_measure_bounded_rates} holds.
 Since $\widehat{Z}^N_\cdot$ converges in probability to $Z^*_\cdot$ in the Skorohod space, by a version of the continuous mapping theorem \cite[Section 5.4]{hoffmann:probability} it follows that  
  $$P\pr{\sup_{t\in[0,T]}\abs{\int_0^t\pr{f(\widehat{Z}^N_s)-f(Z^*_s)}ds}>\varepsilon}\xrightarrow[N\rightarrow\infty]{}0,$$
  where we used that the Skorohod distance for continuous functions is equivalent to the uniform distance.
  Hence, \eqref{eq:occupation_measure_unbounded_rates} is a consequence of triangular inequality. By similar arguments and by Theorem \ref{thm:weak_convergence_bounded_rates}, if $\alpha>0$ we have
  \begin{align*}
   P\pr{\sup_{t\in[0,T]}\abs{\widehat{X}^N_{t}-\widehat{Z}^N_{t}}>\varepsilon}\leq  4\nu + P\pr{\sup_{t\in[0,T]}\abs{\widehat{X}^N_{b,t}-\widehat{Z}^N_{b,t}}>\varepsilon}\xrightarrow[N\rightarrow\infty]{}4\nu.
  \end{align*}
  If $\alpha>0$ then $Z^*_\cdot$ is continuous, therefore by Lemma \ref{lemma:limit_process} we have
  $$P\pr{\sup_{t\in[0,T]}\abs{\widehat{Z}^N_t-Z^*_t}>\varepsilon}\xrightarrow[N\rightarrow\infty]{}0.$$
  The proof is then concluded by the arbitrariness of $\nu$, and the triangular inequality.
 \end{proof}

 \begin{rem}
  Convergence of the processes $\widehat{X}^N_{\cdot\wedge T}$ to the process $Z^*_{\cdot\wedge T}$ in occupation measure is implied by \eqref{eq:occupation_measure_unbounded_rates} \cite[Theorem 4.5]{kallenberg:lectures}.
 \end{rem}
 
 \begin{cor}\label{cor:cor_to_general_case}
  Assume that the hypotheses of Lemma \ref{lemma:limit_process} are satisfied. Then, the difference between the processes $\widehat{X}^N_{\cdot\wedge T}$ and $Z^*_{\cdot\wedge T}$ converges in finite dimensional distribution to 0.
 \end{cor}

 \begin{proof}
  The proof is identical to the proof of Corollary \ref{cor:cor_to_bounded_rate_lemma}. Indeed, from Theorem \ref{thm:main_theorem}, we have that \eqref{eq_to_be_proved_unbounded_rates} holds for any $\varepsilon>0$. Thus, for any finite set of time points $\prg{t_m}_{m=0}^p\subseteq[0,T]$, we have that
  \begin{align*}
   P\pr{\max_{0\leq m\leq p}\abs{\widehat{X}^N_{t_m}-Z^*_{t_m}}>\varepsilon}&=P\pr{\bigcup_{m=0}^p\prg{\abs{\widehat{X}^N_{t_m}-Z^*_{t_m}}>\varepsilon}}\\
   &\leq\sum_{m=0}^p P\pr{\abs{\widehat{X}^N_{t_m}-Z^*_{t_m}}>\varepsilon}\xrightarrow[N\rightarrow\infty]{}0,
  \end{align*}
  and the result follows.
 \end{proof}

 \section{Discussion}
 
 We close by presenting  a collection of examples and remarks.  A particular strength of our approach is that the reduced reaction network  is easily found from the original reaction network and that the reaction rates of the reduced SRN can be found through a simple algebraic procedure. If the definition of intermediate species is relaxed, it might still be possible to find an approximating reduced SRN in concrete cases.  However, a  general technique does not seem to present itself easily. 

 We assume mass-action kinetics unless otherwise specified. If the stoichiometric coefficient of the intermediates were allowed to be different from one, or if different intermediate species were allowed to interact, our results would not be true in general:
 
 \begin{ex}[Relaxing the definition of intermediates, I]
Consider the SRN
  
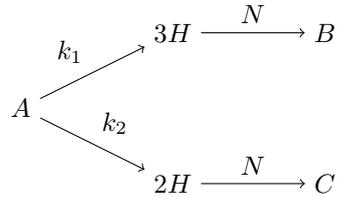
\begin{center}
   \begin{tikzpicture}
    \node[state] (A)  at (1,0)  {$A$};
    \node[state] (3\is) at (3,1)  {$3\is$};
    \node[state] (B)  at (5,1)  {$B$};
    \node[state] (2\is) at (3,-1) {$2\is$};
    \node[state] (C)  at (5,-1) {$C$};
    \path[->] 
     (A)  edge node{$k_1$} (3\is)
     (3\is) edge node{$N$}   (B)
     (A)  edge node{$k_2$} (2\is)
     (2\is) edge node{$N$}   (C);
   \end{tikzpicture}
  \end{center}
  with $\alpha=0$. A single molecule of $\is$ could be trapped as the two reactions $3\is\rightarrow B$ and $2\is\rightarrow A$ compete against each other. Thus, there does not exist an approximation without intermediates as in Theorem \ref{thm:eq_true_for_bounded_rates} or Theorem \ref{thm:main_theorem}. An approximation with no fast species, however, still exists. Since the dynamics of the system changes depending on whether a molecule of $\is$ is present or not, we might introduce two dummy variables $D_1$ and $D_2$ with $D_1+D_2=1$, and $D_1=1$ if and only if no molecules of $\is$ are present. Let $\widehat{p}$ denote the projection onto the space of non-dummy variables. The finite dimensional distributions of $p(X^N_\cdot)$ are approximated by the finite dimensional distributions of $\widehat{p}(Z^N_\cdot)$, where $Z^N_\cdot$ is the process associated with
  
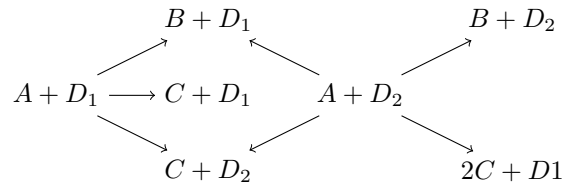
\begin{center}
   \begin{tikzpicture}
    \node[state] (A+D1)  at (1,0)  {$A+D_1$};
    \node[state] (B+D1)  at (3,1)  {$B+D_1$};
    \node[state] (C+D1)  at (3,0)  {$C+D_1$};
    \node[state] (C+D2)  at (3,-1) {$C+D_2$};
    \node[state] (A+D2)  at (5,0)  {$A+D_2$};
    \node[state] (B+D2)  at (7,1)  {$B+D_2$};
    \node[state] (2C+D1) at (7,-1) {$2C+D1$};
    \path[->] 
     (A+D1)  edge node{} (B+D1)
             edge node{} (C+D1)
             edge node{} (C+D2)
     (A+D2)  edge node{} (B+D2)
             edge node{} (C+D2)
             edge node{} (B+D1)
             edge node{} (2C+D1);
   \end{tikzpicture}
  \end{center}
  for a suitable choice of kinetics and with initial conditions $X_0(D_1)=1$ and $X_0(D_2)=0$. A general reduction technique that can deal with examples of this kind is subject of further investigation.  Similar arguments can be made if intermediate species are interacting, for example, if $3H$ and $2H$ are replaced by $H_1+H_2$ and $H_1$, respectively. 
 \end{ex}
 
 \begin{ex}[Relaxing the definition of intermediates, II]
Consider the SRN below with  $\alpha(C)=\alpha(F)=1$ and $\alpha(A)=\alpha(B)=\alpha(D)=\alpha(E)=0$:
 
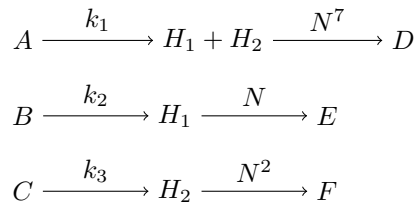
\begin{center}
  \begin{tikzpicture}
   \node[state] (A)       at (3,1)   {$A$};
   \node[state] (\is_1+\is_2) at (5.5,1) {$\is_1+\is_2$};
   \node[state] (B)       at (3,0)   {$B$};
   \node[state] (\is_1)     at (5,0)   {$\is_1$};
   \node[state] (\is_2)     at (5,-1)  {$\is_2$};
   \node[state] (C)       at (3,-1)  {$C$};
   \node[state] (D)       at (8,1)   {$D$};
   \node[state] (E)       at (7,0)   {$E$};
   \node[state] (F)       at (7,-1)  {$F$};
   \path[->] 
    (A)       edge node{$k_1$} (\is_1+\is_2)
    (\is_1+\is_2) edge node{$N^7$} (D)
    (B)       edge node{$k_2$} (\is_1)
    (\is_1)     edge node{$N$}   (E)
    (C)       edge node{$k_3$} (\is_2)
    (\is_2)     edge node{$N^2$} (F);
  \end{tikzpicture}
 \end{center}
Here, a reaction of type $C\rightarrow \is_2$ can occur before a present molecule of $\is_1$ is consumed, leading to the production of $D$ from $\is_1+\is_2\rightarrow D$. It can be shown that the right limit is given by the rescaled process associated with
 \begin{center}
  \begin{tikzpicture}
   \node[state] (A)       at (3,1)   {$A$};
   \node[state] (B)       at (3,0)   {$B$};
   \node[state] (C)       at (7,.5)  {$C$};
   \node[state] (D)       at (5,1)   {$D$};
   \node[state] (E)       at (5,0)   {$E$};
   \node[state] (F)       at (9,.5)  {$F$};
   \path[->] 
    (A) edge node{$k_1$} (D)
    (B) edge node{}      (D)
    (B) edge node{}      (E)
    (C) edge node{$k_3$} (F);
  \end{tikzpicture}
 \end{center}
 where
$$\lambda^N_{B\rightarrow E}(x)=k_2x(B)\frac{N}{N+k_3x(C)},\qquad
   \lambda^N_{B\rightarrow D}(x)=k_2x(B)\frac{k_3x(C)}{N+k_3x(C)}.$$
 If we change the rate constant of $\is_1\rightarrow E$ to $N^2$ and let $\alpha(B)=\alpha(E)=1$, a different reduced SRN is obtained in which a new complex appears:
 \begin{center}
  \begin{tikzpicture}
   \node[state] (A)   at (3,0)   {$A$};
   \node[state] (B+C) at (7.5,0) {$B+C$};
   \node[state] (B)   at (9,0) {$B$};
   \node[state] (C)   at (12,0) {$C$};
   \node[state] (D)   at (5,0)   {$D$};
   \node[state] (E)   at (11,0) {$E$};
   \node[state] (F)   at (14,0) {$F$};
   \path[->] 
    (A)   edge node{$k_1$}                (D)
    (B)   edge node{$k_2$}                (E)
    (C)   edge node{$k_3$}                (F)
    (B+C) edge node[swap]{$\frac{k_2k_3}{N^2}$} (D);
  \end{tikzpicture}
 \end{center}
 \end{ex}

 It would be desirable to state Theorem \ref{thm:main_theorem} in terms of the stronger notion of convergence in probability in the Skorohod space, or at least in terms of the weak convergence in the Skorohod topology. This is done for $\alpha>0$ (cf. Theorems \ref{thm:eq_true_for_bounded_rates} and \ref{thm:main_theorem}), however it cannot be done in general as shown in the next example.

 \begin{ex}[Weak convergence]\label{ex:weak_convergence}
  Consider
  \begin{center}
   \begin{tikzpicture}
    \node[state] (A) at (1,1) {$A$};
    \node[state] (\is) at (3,1) {$\is$};
    \node[state] (B) at (5,1) {$B$};
    \path[->] 
     (A) edge node{$k$} (\is)
     (\is) edge node{$N$} (B);
   \end{tikzpicture}
  \end{center}
  with $\alpha(A)=\alpha(B)=0$, and the limit process $\widehat{Z}^N_\cdot$ associated with the SRN
  \begin{center}
   \begin{tikzpicture}
    \node[state] (A) at (1,1) {$A$};
    \node[state] (B) at (3,1) {$B$};
    \path[->] 
     (A) edge node{$k$} (B);
   \end{tikzpicture}
  \end{center}
  Since the reduced SRN does not depend on $N$, we omit $N$ in the notation. The possible states of $\widehat{Z}_\cdot$ satisfy the conservation law $\widehat{Z}_t(A)+\widehat{Z}_t(B)=M$ for some fixed $M$. In contrast, whenever the reaction $A\rightarrow \is$ occurs in $X^N_\cdot$, for a short amount of time at least, $\widehat{X}^N_t(A)+\widehat{X}^N_t(B)\leq M-1$. The latter situation happens with positive probability, such that
  $$E\prq{\inf_{[0,T]}\pr{\widehat{X}^N_t(A)+\widehat{X}^N_t(B)}}-E\prq{\inf_{[0,T]}\pr{\widehat{Z}_t(A)+\widehat{Z}_t(B)}}\xrightarrow[N\rightarrow\infty]{}c\neq0.$$
  Hence, $\widehat{Z}_{\cdot\wedge T}$ does not provide a weak limit in the Skorohod topology for $\widehat{X}^N_{\cdot\wedge T}$. In fact, in this particular case the sequence of processes $\widehat{X}^N_{\cdot\wedge T}$ cannot have a weak limit in the Skorohod topology, since the sequence of the corresponding distributions $P^N$ is not tight. 
 \end{ex}

 A natural question arising from the results of this paper is whether the reduced reaction network could be used to approximate the limit behaviour of the full model as $t\rightarrow\infty$. Specifically, we want to investigate whether for all Borel sets $A\subset\RR^{\Sp\setminus\In}$, it holds that
 \begin{equation}\label{eq:conv_lim_distribution}
  \lim_{t\rightarrow\infty}P\pr{\widehat{X}^N_t\in A}-\lim_{t\rightarrow\infty}P\pr{\widehat{Z}^N_t\in A}\xrightarrow[N\rightarrow\infty]{}0,
 \end{equation}
 under the hypothesis that the limits exist. The answer is negative, as it is shown with the next example.
  
 \begin{ex}[Limit behaviour in the stochastic setting]
  Consider the following SRN.
  \begin{center}
   \begin{tikzpicture}
    \node[state] (A)     at (1,1)   {$A$};
    \node[state] (\is)   at (3,1.5) {$\is$};
    \node[state] (B)     at (5,1)   {$B$};
    \node[state] (empty) at (7,1)   {$0$};
    \path[->] 
     (A)   edge            node{$k_1$}        (\is)
     (\is) edge            node{$N$}          (B)
     (B)   edge[bend left] node{$k_2$}        (A)
           edge            node{$\lambda(x)$} (empty);
   \end{tikzpicture}
  \end{center}
  Let $\alpha(A)=\alpha(B)=0$, assume that $X^N_0(A)+X^N_0(B)=M$ and $X^N_0(\is)=0$, and let
  $$\lambda(x)=(M-x(A)-x(B))\mathbbm{1}_{(0,\infty)}(x(B)).$$
  The first occurrence of the reaction $B\rightarrow0$ can only take place when $H$ is present. Though it is unlikely for big $N$, there is still a positive probability that this happens, i.e.\ that $B\rightarrow0$ occurs before the reaction $\is\rightarrow B$ takes place. With probability one, all molecules of $B$ will eventually be consumed, and the limit distribution of the above SRN is therefore concentrated on the state $0$.
   
  The  SRN satisfies the assumptions of Theorem \ref{thm:eq_true_for_bounded_rates} and those of Theorem \ref{thm:main_theorem}. The reduced reaction network is given by
  
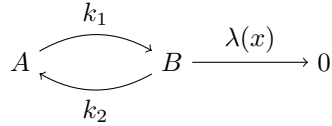
\begin{center}
   \begin{tikzpicture}
    \node[state] (A)     at (1,1) {$A$};
    \node[state] (B)     at (3,1) {$B$};
    \node[state] (empty) at (5,1) {$0$};
    \path[->] 
     (A) edge[bend left=30] node{$k_1$} (B)
     (B) edge[bend left=30] node{$k_2$} (A)
         edge node{$\lambda(x)$} (empty);
   \end{tikzpicture}
  \end{center}
  where the initial conditions are the same as in the bigger model. Since in the reduced SRN $\lambda(Z^N_t)=0$ whenever $Z^N_t(A)+Z^N_t(B)=M$, then the reaction $B\to 0$ never occurs. This implies that  the reduced SRN is equivalent to
  
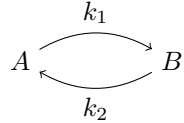
\begin{center}
   \begin{tikzpicture}
    \node[state] (A)     at (1,1) {$A$};
    \node[state] (B)     at (3,1) {$B$};
    \path[->] 
     (A) edge[bend left=30] node{$k_1$} (B)
     (B) edge[bend left=30] node{$k_2$} (A);
   \end{tikzpicture}
  \end{center}
  The limit distribution of the above SRN is concentrated on the set $\prg{x\colon x(A)+x(B)=M}$. Therefore it is clear that (\ref{eq:conv_lim_distribution}) does not hold in this case. However, the limit distribution of the latter SRN approximates the \emph{quasi-stationary distribution} of the original SRN when $N$ tends to infinity, if we condition on the event that the reaction $B\to 0$ has not taken place; see for example  \cite{anderson:product-form,anderson:ACR} for a discussion on stationary and quasi-stationary distributions in reaction network theory.
 \end{ex}
 
 \section{Proof of Theorem \ref{thm:model_comparison} and Proposition \ref{prop:mu_suff_ass}}\label{sec:proofs}
 
 This section is devoted to prove Theorem \ref{thm:model_comparison} and Proposition \ref{prop:mu_suff_ass}. First, recall the transition rate matrix \eqref{eq:transition_matrix}. Consider the continuous time Markov chain $C^x$ with state space $\SR\sqcup\In\sqcup\FP$ (disjoint union) and transition rate matrix given by 
 
 {\renewcommand\arraystretch{1.9}
  $$Q(x)=\left[\begin{array}{C|C|C}
                 Q_{\sr, \sr}^x & Q_{\sr, \ii}^x & 0                     \\ \hline
                 0                        & Q_{\ii, \ii}   & Q_{\ii, \fp} \\ \hline
                 0                        & 0                        & 0                     \\
                \end{array}\right],$$}
 where $Q_{\sr, \ii}^x$ is defined by
 $$q_{i\ell}^x=\lambda_{i\ell}(x)$$
 for $i\in\sr$, $\ell\in\ii$, and 
 $$Q_{\sr, \sr}^x\df\diag(-Q_{\sr, \ii}^x e),$$
 where $e$ denotes a vector of suitable length with all entries  1. Given a matrix $M$, we denote by $M_i$ its $i$-th row. Note that the matrix
 {\renewcommand\arraystretch{1.9}
  $$\La^x_i=\left[\begin{array}{D|D}
                 -Q_{\ii, \ii}         & -Q_{\ii, \fp}e         \\ \hline
                 -(Q_{\sr, \ii}^x)_i & (Q_{\sr, \ii}^x)_ie \\
                \end{array}\right]$$}
 is the transposed Laplacian matrix of the graph $\mathcal{G}_i^x$ defined in (\ref{eq:G_i^x}) (row sums are zero). Let $D^x(\cdot)$ denote the discrete time Markov chain embedded in $C^x(\cdot)$ and let $\PP(x)$ be the corresponding transition probability matrix of $D^x$. For any $i\in\sr$, let
  $$P^x_i\pr{\cdot}=P\pr{\cdot | D^x(0)=y_i}, \quad \text{and}\quad  E^x_i\prq{\cdot}=E\prq{\cdot | D^x(0)=y_i}.$$
 Moreover, let
 \begin{align*}
    \outx{i}&=\sum_{\ell\in\iii}\lambda_{i\ell}(x)=(Q_{\sr, \ii}^x)_ie,\\
    \out{\ell}&=\sum_{\ell^\prime\in\ii}k_{\ell\ell^\prime}+\sum_{j\in\fp}k_{\ell j}=\left[Q_{\ii, \ii} \,|\, Q_{\ii, \fp}\right]_\ell e.
 \end{align*}
 We have the following result:
 \begin{lem}\label{lem:mu_expectation}
  For all $\ell\in\ii$, $i\in\sr$ and  $x\in\RR^{\Sp\setminus\In}_{\ge 0}$, 
  \begin{equation}\label{eq:lem_mu_expectation}
   \mu_{i\ell}(x)=\frac{\outx{i}}{\out{\ell}}\sum_{n\geq 1}(\PP(x)^n)_{i\ell}<\infty.
  \end{equation}
  In particular, we have
  $$\mu_{i\ell}(x)=\frac{\outx{i}}{\out{\ell}}E^x_i\prq{\#\text{visits of $D^x(\cdot)$ to }\is_{\ell}}.$$
 \end{lem}
 \begin{proof}
 We have
 $$E^x_i\prq{\#\text{visits of $D^x(\cdot)$ to }\is_{\ell}}=E^x_i\prq{\sum_{n\geq 1}\mathbbm{1}_{\prg{\is_{\ell}}}D^x(n)}=\sum_{n\geq 1}P^x_i\pr{D^x(n)=\is_{\ell}}=\sum_{n\geq 1}(\PP(x)^n)_{i\ell}.$$
 Therefore, since every intermediate species is a transient state in $D^x$,
 $$\sum_{n\geq 1}(\PP(x)^n)_{i\ell}=E^x_i\prq{\#\text{visits of $D^x(\cdot)$ to }\is_{\ell}}<\infty.$$
 Thus, we only need to prove (\ref{eq:lem_mu_expectation}).
 The matrix $\PP(x)$ has the following block structure:
 {\renewcommand\arraystretch{1.9}
  $$\PP(x)=\left[\begin{array}{C|C|C}
                0 & \PP_{\sr,\ii}^x & 0 \\ \hline
                0 & \PP_{\ii,\ii}  & \PP_{\ii,\sr} \\ \hline
                0 & 0 & I \\
               \end{array}\right].$$
  Thus, we have
  $$\PP(x)^n=\left[\begin{array}{C|c|C}
                 0 & \PP_{\sr,\ii}^x\PP_{\ii,\ii}^{n-1} & * \\ \hline
                 0 & \PP_{\ii,\ii}^n & * \\ \hline
                 0 & 0 & * \\
                \end{array}\right].$$}
  Since for any $\ell,\ell^\prime\in\ii$,
  $$\pr{\sum_{n\geq 0}\PP_{\ii,\ii}^n}_{\!\!\ell\ell^\prime}=\Econ{\#\text{visits of $D^x(\cdot)$ to }\is_{\ell^\prime}}{D^x(1)=\is_\ell}<\infty,$$
  we have that $\sum_{n\geq 0}\PP_{\ii,\ii}^n$ is well defined and
  $$\sum_{n\geq 0}\PP_{\ii,\ii}^n=(I-\PP_{\ii,\ii})^{-1}.$$
  Therefore
  \begin{equation}\label{eq:equivalence1_det&stoch}
   \sum_{n\geq 1}(\PP(x)^n)_{i\ell}=\sum_{n\geq 0}(\PP_{\sr,\ii}^x\PP_{\ii,\ii}^n)_{i\ell}=\pr{\PP_{\sr,\ii}^x(I-\PP_{\ii,\ii})^{-1}}_{i\ell}.
  \end{equation}
  Assume $\outx{i}\not=0$. Consider the graph $\widetilde{\mathcal{G}}_i^x$ with the same nodes and edges as $\mathcal{G}_i^x$ and normalized labels
  $$\is_\ell\xrightarrow{\displaystyle\frac{k_{\ell\ell^\prime}}{\out{\ell}}}\is_{\ell^\prime}\quad,\quad
  \is_\ell\xrightarrow{\displaystyle\frac{\sum_{j\in\fp} k_{\ell j}}{\out{\ell}}}\star\quad,\quad
  \star\xrightarrow{\displaystyle\frac{\lambda_{i\ell}(x)}{\outx{i}}}\is_\ell.$$
  The transpose of the Laplacian matrix of the graph $\widetilde{\mathcal{G}}_i^x$ is given by
  {\renewcommand\arraystretch{1.9}
  $$\widetilde{\La}_i^x=\left[\begin{array}{c|c}
                               I-\PP_{\ii,\ii} & -\PP_{\ii,\fp}e \\ \hline
                               -(\PP_{\sr,\ii}^x)_i & 1 \\
                              \end{array}\right].$$}
  Given a matrix $M$, denote by $M_{(i,j)}$ the matrix obtained by $M$ eliminating the $i$-th row and the $j$-th column. We have that
  \begin{align*}
   \pr{\PP_{\sr,\ii}^x(I-\PP_{\ii,\ii})^{-1}}_{i\ell}&=\sum_{\ell^\prime\in\ii}(\PP_{\sr,\ii}^x)_{i\ell^\prime}(I-\PP_{\ii,\ii})^{-1}_{\ell^\prime\ell}
   =\sum_{\ell^\prime\in\ii}(\PP_{\sr,\ii}^x)_{i\ell^\prime}\frac{(-1)^{\ell+\ell^\prime}\det(I-\PP_{\ii,\ii})_{(\ell,\ell^\prime)}}{\det(I-\PP_{\ii,\ii})}\\
   &=(-1)^{\ell+\#\In +1}\frac{\det (\widetilde{\La}_i^x)_{(\ell,\#\In +1)}}{\det(\widetilde{\La}_i^x)_{(\#\In +1,\#\In +1)}} 
   =\frac{\det (\widetilde{\La}_i^x)_{(\ell,\ell)}}{\det(\widetilde{\La}_i^x)_{(\#\In +1,\#\In +1)}}\\
   &=\frac{\out{\ell}}{\outx{i}}\frac{\det (\La_i^x)_{(\ell,\ell)}}{\det (\La_i^x)_{(\#\In +1,\#\In +1)}} 
 =\frac{\out{\ell}}{\outx{i}}\frac{\sum_{\zeta\in\Theta_{i,x}(\is_\ell)}w(\zeta)}{\sum_{\zeta\in\Theta_{i,x}(\star)}w(\zeta)} 
 =\frac{\out{\ell}}{\outx{i}}\mu_{i\ell}(x),
  \end{align*}
  where the second equality follows from the co-factor expansion of the determinant, the third from the Laplace expansion and the fourth equality follows from the fact that the last column of the Laplacian matrix is equal to minus the sum of the other columns. The second-last equality follows from the Matrix-Tree theorem \citep{tutte:matrix_tree}. Thus, from (\ref{eq:equivalence1_det&stoch}) it follows that (\ref{eq:lem_mu_expectation}) holds.
  If $\outx{i}=0$, then $\mu_{i\ell}=0$ for all $\ell\in\ii$. Thus, (\ref{eq:lem_mu_expectation}) still holds and the proof is concluded.
 \end{proof}

 The proof of Theorem \ref{thm:model_comparison} follows from Lemma \ref{lem:mu_expectation}.
 
 \begin{proof}[Proof of Theorem \ref{thm:model_comparison}]
  We have to prove that for any fixed $x\in\RR^{\Sp\setminus\In}_{\ge 0}$,
  \begin{equation}\label{eq:eq_to_prove_model_comparison}
   \sum_{\ell\in\ii} \pi_{\ell j}\lambda_{i\ell}(x)=\sum_{\ell\in\ii} k_{\ell j}\mu_{i\ell}(x).
  \end{equation}
  Note that
  \begin{align*}
   \sum_{\ell\in\ii} \pi_{\ell j}\lambda_{i\ell}(x)&=\outx{i}\sum_{\ell\in\ii} P^x_i\prcon{\lim_{n\rightarrow\infty}D^x(n)=y_j}{D^x(1)=\is_\ell}P^x_i\pr{D^x(1)=\is_\ell}\\
   &=\outx{i}P^x_i\pr{\lim_{n\rightarrow\infty}D^x(n)=y_j}=\outx{i}\sum_{n\geq 1}\sum_{\ell\in\ii} P^x_i\pr{D^x(n)=\is_\ell\,,\,D^x(n+1)=y_j}\\
   &=\outx{i}\sum_{\ell\in\ii}\frac{k_{\ell j}}{\out{\ell}}\sum_{n\geq 1}(\PP(x)^n)_{i\ell}.
 \end{align*}
 Therefore, (\ref{eq:eq_to_prove_model_comparison}) follows from Lemma \ref{lem:mu_expectation}.
 \end{proof}

 To prove Proposition \ref{prop:mu_suff_ass}, we make the dependence on $N$ explicit.
 
 \begin{proof}[Proof of Proposition \ref{prop:mu_suff_ass}]
  From Lemma \ref{lem:mu_expectation} we have that
  $$\mu^N_{i\ell}(x)=\frac{\outxN{i}}{\outN{\ell}}E^{N,x}_i\prq{\#\text{visits of $D^{N,x}(\cdot)$ to }\is_{\ell}}$$
  for $x\in \RR^{\Sp\setminus\In}_{\ge 0}$.
  Denote by $T^N_\ell$ the random variable distributed as the time until consumption of a molecule of $\is_\ell$. Its distribution is exponential with parameter $\outN{\ell}$. Note that the Markov chain $D^{N,x}(\cdot)$ is distributed as the discrete time Markov chain embedded in $C^N_1(\cdot)$, and $D^{N,x}(0)$ denotes the initial reactant setting off the chain $C^N_1(\cdot)$. For any $j\in\fp$, we have
  \begin{align*}
   \outaxN{i}\Econ{\tau_1^N}{D^{N,N^\alpha x}(0)=y_i}&=\outaxN{i}\sum_{\ell\in\ii}E^{N,N^\alpha x}_i\prq{\#\text{visits of $D^{N,N^\alpha x}(\cdot)$ to }\is_{\ell}}E\prq{T^N_\ell}\\
   &=\sum_{\ell\in\ii}\mu^N_{i\ell}(N^\alpha x).
  \end{align*}
 Furthermore,
  \begin{align*}
   \outaxN{i}\Econ{\tau_1^N}{D^{N,N^\alpha x}(0)=y_i}&=\\
   &\hspace*{-48pt}=\sum_{\ell\in\ii}\Econ{\tau_1^N}{D^{N,N^\alpha x}(1)=\is_\ell}P\prcon{D^{N,N^\alpha x}(1)=\is_\ell}{D^{N,N^\alpha x}(0)=y_i}\outaxN{i}\notag\\
   &\hspace*{-48pt}=\sum_{\ell\in\ii}\Econ{\tau_1^N}{D^{N,N^\alpha x}(1)=\is_\ell}\lambda^N_{i\ell}(N^\alpha x)\notag\\
   &\hspace*{-48pt}=\sum_{\ell\in\ii}\Econ{N^{\beta_{i\ell}}\tau_1^N}{D^{N,N^\alpha x}(1)=\is_\ell}N^{-\beta_{i\ell}}\lambda^N_{i\ell}(N^\alpha x).
  \end{align*}
  In particular,
  \begin{equation}\label{eq:sum_exp_in_proof_prop_mu}
   \sum_{\ell\in\ii}N^{-\alpha^*}\mu^N_{i\ell}(N^\alpha x)=\sum_{\ell\in\ii}N^{\beta_{i\ell}-\alpha^*}\Econ{\tau_1^N}{D^{N,N^\alpha x}(1)=\is_\ell}N^{-\beta_{i\ell}}\lambda^N_{i\ell}(N^\alpha x).
  \end{equation}
  Therefore, \eqref{eq:mu_tend_to_zero} holds if and only if the right-hand side of \eqref{eq:sum_exp_in_proof_prop_mu} tends to zero as $N\rightarrow\infty$. By Assumption \ref{ass:big_with_N}(\ref{ass:ratefct}) we have
  $$N^{-\beta_{i\ell}}\lambda^N_{i\ell}(N^\alpha x)\xrightarrow[N\rightarrow\infty]{}\lambda_{i\ell}(x),$$
  where $\lambda_{i\ell}$ is a non-null function. It follows that the right-hand side of \eqref{eq:sum_exp_in_proof_prop_mu} tends to zero as $N\rightarrow\infty$ if and only if, for any $i\in\sr,\ell\in\ii$, such that $y_i\rightarrow \is_\ell\in\R$, and for any $j\in\fp$
  $$N^{\beta_{i\ell}-\alpha^*}\Econ{\tau_1^N}{D^{N,N^\alpha x}(1)=\is_\ell}\xrightarrow[N\rightarrow\infty]{}0.$$
  The latter is equivalent to
  \begin{equation}\label{eq:last_eq_proof_prop_mu}
   N^{\beta_{i\ell}-\alpha^*}\Econ{\tau_1^N}{C^N_1(t^N_1)=\is_\ell}\xrightarrow[N\rightarrow\infty]{}0.
  \end{equation}
  By the definition of $\beta^*_{\ell}$ and $a_\ell$, the latter implies
  $$N^{\beta^*_{\ell}-a_\ell}\Econ{\tau_1^N}{C^N_1(t^N_1)=\is_\ell}\xrightarrow[N\rightarrow\infty]{}0$$
  for any $\ell\in\bigcup_{i\in\sr}\iii$, which is what we wanted to prove. If $\alpha^*_j=\alpha^*_{j^\prime}$ for any $j,j^\prime\in\fp$, then $a_\ell=\alpha^*$ for any $\ell\in\ii$. Therefore, \eqref{eq:last_eq_proof_prop_mu} for any $j\in\fp$ is equivalent to \eqref{eq:exp_lifespan_going_to_zero}. The proof is thus concluded.
 \end{proof}
 
 \section{Proof of Theorems \ref{thm:eq_true_for_bounded_rates} and \ref{thm:weak_convergence_bounded_rates}}\label{sec:proofs_bounded_rates}
 
 In this section Theorems \ref{thm:eq_true_for_bounded_rates} and \ref{thm:weak_convergence_bounded_rates} are proven. To this aim, instead of assuming \eqref{eq:probability_limit_assumption} and \eqref{eq:single_scale_limit} in Assumption \ref{ass:big_with_N}(\ref{ass:single_scale_system}), we make use of the weaker conditions \eqref{eq:limsup_limited_hypothesis} and \eqref{eq:single_scale}.  Throughout this section, whenever $t$ is written it is implicitly assume that $t\in[0,T]$. We also use the notation
 $$\norm{x}_{\infty}=\max_{S\in\Sp}\abs{x(S)}\quad\text{for }x\in \RR^{\Sp}.$$
 By the equivalence of norms in $\RR^{\Sp}$, we have that there exists $\theta>0$, such that
 $$\abs{x}\leq \theta \norm{x}_\infty\quad\forall x\in \RR^{\Sp}.$$
 Let $D_\infty(h)$ be the disc of radius $h$ in $\RR_{\geq0}^{\Sp\setminus \In}$ centred in the origin, with respect to the measure $\norm{\cdot}_\infty$, and let $D^C_\infty(r)$ be its complementary.\\
 We start by stating a lemma.
 
 \begin{lem}\label{lem:difference_M}
  Assume the assumptions of Theorem \ref{thm:eq_true_for_bounded_rates} hold. Then,
  $$\sup_{t\in[0,T]}N^{-\alpha^*_j}E\prq{M_{i\ell j}^N(t)-\overline{M}_{i\ell j}^N(t)}\xrightarrow[N\to\infty]{}0$$
 \end{lem}
 \begin{proof}
  Remember that $t_{i\ell j,n}^N$ is the time of the $n$-th jump of $M_{i\ell j}^N(t)$, and $\tau_{i\ell j,n}^N$ the life time of the corresponding chain of intermediates. Note that by \eqref{eq:limsup_limited_hypothesis} we have
  $$0\leq \sup_{t\in[0,T]}N^{-\alpha^*_j}E\prq{M_{i\ell j}^N(t)-\overline{M}_{i\ell j}^N(t)}\leq \sup_{t\in[0,T]}N^{-\alpha^*_j}E\prq{M_{i\ell j}^N(t)}\leq \pi^N_{\ell j}N^{\beta^*_\ell-\alpha^*_j}B_{i\ell}t\leq c_{\ell j}B_{i\ell}t.$$
  This implies that the sequence $\sup_{t\in[0,T]}N^{-\alpha^*_j}E\prq{M_{i\ell j}^N(t)-\overline{M}_{i\ell j}^N(t)}$ is contained in a compact set, and it follows that to prove the lemma it is sufficient to show that all the accumulation points of the sequence are 0. To this aim, fix an accumulation point $l$ and consider a subsequence $N_h$ such that
  $$\sup_{t\in[0,T]}N_h^{-\alpha^*_j}E\prq{M_{i\ell j}^{N_h}(t)-\overline{M}_{i\ell j}^{N_h}(t)}\xrightarrow[h\to\infty]{}l.$$
  First, assume that
  $$\liminf_{h\to\infty}\pi^{N_h}_{\ell j}N_h^{\beta^*_\ell-\alpha^*_j}=0,$$
  and let $N_{h_m}$ a subsequence such that
  $$\lim_{m\to\infty}\pi^{N_{h_m}}_{\ell j}N_{h_m}^{\beta^*_\ell-\alpha^*_j}=0.$$
  In this case,
  $$0\leq l=\lim_{m\to\infty}\sup_{t\in[0,T]}N_{h_m}^{-\alpha^*_j}E\prq{M_{i\ell j}^{N_{h_m}}(t)-\overline{M}_{i\ell j}^{N_{h_m}}(t)}\leq \lim_{m\to\infty}\pi^{N_{h_m}}_{\ell j}N_{h_m}^{\beta^*_\ell-\alpha^*_j}B_{i\ell}t=0,$$
  which proves $l=0$. Now, assume that
  $$\liminf_{h\to\infty}\pi^{N_h}_{\ell j}N_h^{\beta^*_\ell-\alpha^*_j}=\delta>0,$$
  and fix $0<\varepsilon<\delta t$. For convenience, denote
  $$\sigma^{\varepsilon h}_{j\ell}=\frac{\varepsilon N_h^{\alpha^*_j-\beta^*_\ell}}{\pi^{N_h}_{\ell j}}.$$
  We have
  \begin{align*}
   E\prq{M_{i\ell j}^{N_h}(t)-\overline{M}_{i\ell j}^{N_h}(t)}= & E\prq{\sum_{n=1}^{M_{i\ell j}^{N_h}(t)} \mathbbm{1}_{\left[t-t_{i\ell j,n}^{N_h},\infty\right)}(\tau_{i\ell j,n}^{N_h})}\\
   \leq & E\prq{\sum_{n=1}^{M_{i\ell j}^{N_h}\pr{t-\sigma^{\varepsilon h}_{j\ell}}}\mathbbm{1}_{\left[\sigma^{\varepsilon h}_{j\ell},\infty\right)}(\tau_{i\ell j,n}^{N_h})}+E\prq{\sum_{n=M_{i\ell j}^{N_h}\pr{t-\sigma^{\varepsilon h}_{j\ell}}+1}^{M_{i\ell j}^{N_h}(t)} 1}\\
   \leq & E\prq{\sum_{n=1}^{M_{i\ell j}^{N_h}(t)}\mathbbm{1}_{\left[\sigma^{\varepsilon h}_{j\ell},\infty\right)}(\tau_{i\ell j,n}^{N_h})}+E\prq{M_{i\ell j}^{N_h}(t)-M_{i\ell j}^{N_h}\pr{t-\sigma^{\varepsilon h}_{j\ell}}}.
  \end{align*}
  Thus, using \eqref{eq:def_Y_ilj} and \eqref{eq:bounded_rates} we obtain
  \begin{align*}
   E\prq{M_{i\ell j}^{N_h}(t)-\overline{M}_{i\ell j}^{N_h}(t)}&\leq E\prq{\sum_{n=1}^{Y_{i\ell j}(tB_{i\ell}\pi^{N_h}_{\ell j}N_h^{\beta_{i\ell}})}\mathbbm{1}_{\left[\sigma^{\varepsilon h}_{j\ell},\infty\right)}(\tau_{i\ell j,n}^N)}+\varepsilon N^{\alpha^*_j-\beta^*_\ell}N^{\beta^*_\ell}B_{i\ell}\\
   &\leq tp_{\ell j}^\varepsilon(N_h)N_h^{\beta^*_\ell}B_{i\ell}\pi^{N_h}_{\ell j}+\varepsilon N_h^{\alpha^*_j}B_{i\ell}\leq B_{i\ell}N_h^{\alpha^*_j}(t\pi^{N_h}_{\ell j}N_h^{\beta^*_\ell-\alpha^*_j}p_{\ell j}^\varepsilon(N_h)+\varepsilon),
  \end{align*}
  where $p_{\ell j}^\varepsilon(N)$ is as defined in \eqref{eq:p_ell_j}. By \eqref{eq:lifespan_going_to_zero_relaxed} and the arbitrariness of $\varepsilon>0$, the latter implies that
  $$\sup_{t\in[0,T]}N_h^{-\alpha^*_j}E\prq{M_{i\ell j}^{N_h}(t)-\overline{M}_{i\ell j}^{N_h}(t)}\xrightarrow[h\to\infty]{}0,$$
  which implies that $l=0$ and concludes the proof.
\end{proof}
 
 \begin{proof}[Proof of Theorem \ref{thm:eq_true_for_bounded_rates}]
 Let the process $\widehat{W}_\cdot^N$ be defined as in (\ref{eq:W^N_t}) and, for any fixed $t$, let $\Delta_t^N=\norm{\widehat{X}^N_t-\widehat{W}^N_t}_\infty$.
 Then, we have
 \begin{align*}
   E\prq{\abs{\widehat{X}^N_t-\widehat{Z}^N_t}}&\leq E\prq{\abs{\widehat{W}^N_t-\widehat{Z}^N_t}+\abs{\widehat{X}^N_t-\widehat{W}^N_t}}\\
   &\hspace*{-60pt}\leq\sum_{r\in\R^1}\abs{N^{-\alpha}\xi_r}  E\prq{\abs{Y_r\pr{\int_0^t \lambda^N_r(X^N_s)ds}-Y_r\pr{\int_0^t \lambda^N_r(Z^N_s)ds}}}+\\
   &\hspace*{-60pt}\quad+\sum_{i\in\sr}\sum_{j\in\fp}\abs{N^{-\alpha}\pr{y_j-y_i}}\sum_{\ell\in\iii} E\prq{\abs{Y_{i\ell j}\pr{\int_0^t \pi_{\ell j}^N\lambda^N_{i\ell}(X^N_s)ds}-Y_{i\ell j}\pr{\int_0^t \pi_{\ell j}^N\lambda^N_{i\ell}(Z^N_s)ds}}}+\\
   &\hspace*{-60pt}\quad+E\prq{\abs{\widehat{X}^N_0-\widehat{Z}^N_0}}+\theta E\prq{\Delta_t^N}\\
   &\hspace*{-60pt}=\sum_{r\in\R^1}\abs{N^{-\alpha+\beta_r}\xi_r}  E\prq{N^{-\beta_r}\abs{Y_r\pr{\int_0^t\lambda^N_r(X^N_s)ds}-Y_r\pr{\int_0^t \lambda^N_r(Z^N_s)ds}}}+\\
   &\hspace*{-60pt}\quad+\sum_{i,j,\ell}\abs{N^{-\alpha+\beta_{i\ell}}\pr{y_j-y_i}} E\prq{N^{-\beta_{i\ell}}\abs{Y_{i\ell j}\pr{\int_0^t \pi_{\ell j}^N\lambda^N_{i\ell}(X^N_s)ds}-Y_{i\ell j}\pr{\int_0^t \pi_{\ell j}^N\lambda^N_{i\ell}(Z^N_s)ds}}}+\\
   &\hspace*{-60pt}\quad+E\prq{\abs{\widehat{X}^N_0-\widehat{Z}^N_0}}+\theta E\prq{\Delta_t^N}.
 \end{align*}
 For any reaction $r\colon y_i\rightarrow y_j\in\R^*$, let
 $$\alpha^*_{ij}=\alpha_r^*=\min_{S\colon \xi_r(S)\neq0}\alpha(S).$$
 Then,
 \begin{align*}
    E\prq{\abs{\widehat{X}^N_t-\widehat{Z}^N_t}}
   \leq&\sum_{r\in\R^1}N^{-\alpha_r^*+\beta_r}\abs{\xi_r} E\prq{\int_0^tN^{-\beta_r}\abs{\lambda^N_r(N^\alpha\widehat{X}^N_s)-\lambda^N_r(N^\alpha\widehat{Z}^N_s)}ds}+\\
   &+\sum_{i,j,\ell} \pi_{\ell j}^N N^{-\alpha_{ij}^*+\beta_{i\ell}}\abs{y_j-y_i} E\prq{\int_0^t N^{-\beta_{i\ell}}\abs{\lambda^N_{i\ell}(N^\alpha\widehat{X}^N_s)-\lambda^N_{i\ell}(N^\alpha\widehat{Z}^N_s)}ds}+\\
   &+E\prq{\abs{\widehat{X}^N_0-\widehat{Z}^N_0}}+\theta E\prq{\Delta_t^N}.
 \end{align*}
 To control the left side, we aim to substitute the functions $\lambda_r^N(\cdot)$ with their limits $\lambda_r(\cdot)$ (Assumption \ref{ass:big_with_N}(\ref{ass:ratefct})). To meet our goal, we first argue that the processes $\widehat{X}^N_\cdot$ and $\widehat{Z}^N_\cdot$ are bounded with high probability.
 
 Let $S\in\Sp\setminus\In$. By substituting the rate functions $\lambda_r^N(\cdot)$ with their upper bounds $N^{\beta_r}B_r$  in (\ref{eq_Z}), and using $t<T$, we obtain that $\widehat{Z}_t^N(S)$ is bounded from above by
 $$\widehat{Z}_0^N(S)+\sum_{r\in\R^1_S}\abs{\xi_r(S)}N^{-\alpha(S)}Y_r(N^{\beta_r}B_r T)+\sum_{i\in\sr}\sum_{j\in\fp}\abs{y_j(S)-y_i(S)}\sum_{\ell\in\iii} N^{-\alpha(S)}Y_{i\ell j}\pr{\pi^N_{\ell j}N^{\beta_{i\ell}}B_{i\ell}T},$$
 where $\R^1_S$ is defined according to \eqref{eq:R_S}. Using the assumptions (\ref{eq:limsup_limited_hypothesis}), (\ref{eq:single_scale}) and the Law of Large Numbers for Poisson processes to control the above expression for $\alpha(S)>0$, we obtain that, for any $\nu>0$, there exists $\Upsilon_\nu^\prime>0$, such that
 $$\limsup_{N\rightarrow\infty}P\pr{\sup_{t\in[0,T]}\norm{\widehat{Z}_t^N}_\infty>\Upsilon_\nu^\prime}<\nu.$$
 Let $\Upsilon_\nu$ be as in \eqref{eq:hypothesis_O_1_a} and let $\Upsilon_\nu^{\prime\prime}=\Upsilon_\nu\vee\Upsilon_\nu^\prime$. Then, if $N$ is large enough,
 \begin{equation}\label{eq:bound_1_proof_prop}
  P\pr{\sup_{t\in[0,T]}\pr{\norm{\widehat{X}_t^N}_\infty\vee\norm{\widehat{Z}_t^N}_\infty}>\Upsilon_{\nu}^{\prime\prime}}<3\nu.
 \end{equation}
 By Assumption \ref{ass:big_with_N}(\ref{ass:ratefct}) we have that 
 $$N^{-\beta_r}\lambda_r^N(N^{\alpha}x)\xrightarrow[N\rightarrow\infty]{}\lambda_r(x)\quad\forall r\in \R^0$$
 uniformly on compact sets. In particular, for any $\nu>0$,
 $$o_{\nu}(N)\df\sup_{x\in D_\infty(\Upsilon_\nu^{\prime\prime})}\abs{N^{-\beta_r}\lambda_r^N(N^{\alpha}x)-\lambda_r(x)}\xrightarrow[N\rightarrow\infty]{}0.$$
 Note that for any $\nu>0$ and $x\in\RR^{\Sp\setminus \In}_{\ge 0}$,
 \begin{equation}\label{eq:bound_2_proof_prop}
  \abs{N^{-\beta_r}\lambda_r^N(N^{\alpha}x)-\lambda_r(x)}\leq o_{\nu}(N)\mathbbm{1}_{D_\infty(\Upsilon_\nu^{\prime\prime})}(x)+2B_r\mathbbm{1}_{D^C_\infty(\Upsilon_\nu^{\prime\prime})}(x).
 \end{equation} 
 Using (\ref{eq:bound_1_proof_prop}) and (\ref{eq:bound_2_proof_prop}) we have
 \begin{align*}
    E\prq{\abs{\widehat{X}^N_t-\widehat{Z}^N_t}}
   \leq&\sum_{r\in\R^1}N^{-\alpha_r^*+\beta_r}\abs{\xi_r} \pr{  E\prq{\int_0^t\abs{\lambda_r(\widehat{X}^N_s)-\lambda_r(\widehat{Z}^N_s)}ds} +2o_{\nu}(N)t+12B_r\nu t}+\\
   &\hspace*{-18pt}+\sum_{i,j,\ell} \pi_{\ell j}^N N^{-\alpha_{ij}^*+\beta_{i\ell}}\abs{y_j-y_i}\pr{ E\prq{\int_0^t \abs{\lambda_{i\ell}(\widehat{X}^N_s)-\lambda_{i\ell}(\widehat{Z}^N_s)}ds}+2o_{\nu}(N)t+12B_{i\ell}\nu t}+\\
   &+E\prq{\abs{\widehat{X}^N_0-\widehat{Z}^N_0}}+\theta E\prq{\Delta_t^N}\\
   \leq& \Psi_1 \int_0^t E\prq{\abs{\widehat{X}^N_s-\widehat{Z}^N_s}}ds + \Psi_2 o_{\nu}(N)t + \Psi_3\nu t + E\prq{\abs{\widehat{X}^N_0-\widehat{Z}^N_0}} + \theta E\prq{\Delta_t^N}
 \end{align*}
 for some positive constants $\Psi_1$, $\Psi_2$, $\Psi_3>0$, independent of $\nu$. In the last inequality we made use of \eqref{eq:limsup_limited_hypothesis} and \eqref{eq:single_scale}, as well as the hypothesis that $\lambda_r$ is Lipschitz for any $r\in\R^0$.
 
 To prove (\ref{eq_expectation_goes_to_zero}), we only need to show that $\sup_{t\in[0,T]}E\prq{\Delta_t^N}\rightarrow0$ for $N\rightarrow\infty$. Indeed if this holds, then by the Gronwall inequality applied to the function $\sup_{t\in[0,T]}E\prq{\abs{\widehat{X}^N_t-\widehat{Z}^N_t}}$ we have
 $$ \sup_{t\in[0,T]}E\prq{\abs{\widehat{X}^N_t-\widehat{Z}^N_t}}\leq \pr{\Psi_2 o_{\nu}(N)T + \Psi_3\nu T+E\prq{\abs{\widehat{X}^N_0-\widehat{Z}^N_0}}+\theta\sup_{t\in[0,T]} E\prq{\Delta_t^N}}e^{\Psi_1 T},$$
 which for $N\to\infty$ tends to $\Psi_3\nu Te^{\Psi_1 T}$. By the arbitrariness of $\nu$ this leads to
 $$ \sup_{t\in[0,T]}E\prq{\abs{\widehat{X}^N_t-\widehat{Z}^N_t}}\xrightarrow[N\rightarrow\infty]{}0,$$
 and we are done. 
 To prove that $\sup_{t\in[0,T]}E\prq{\Delta_t^N}\rightarrow0$ for $N\rightarrow\infty$, it first follows from (\ref{eq:X_t^N_hat}) and (\ref{eq:W^N_t}) that
 \begin{equation}\label{eq:delta_t_minorization}
  \Delta_t^N=\norm{N^{-\alpha}\sum_{j\in\fp}y_j\sum_{i\in\sr}\sum_{\ell\in\iii}\prbig{M_{i\ell j}^N(t)-\overline{M}_{i\ell j}^N(t)}}_\infty\leq \sum_{i,\ell,j} \norm{y_j}_\infty N^{-\alpha^*_j}\prbig{M_{i\ell j}^N(t)-\overline{M}_{i\ell j}^N(t)}.
 \end{equation}
 Therefore, by Lemma \ref{lem:difference_M} and \eqref{eq:delta_t_minorization}, we have that $\sup_{t\in[0,T]}E\prq{\Delta_t^N}\rightarrow0$ for $N\rightarrow\infty$, which concludes the proof of the first part of the statement. Equation \eqref{eq_to_be_proved} is implied by \eqref{eq_expectation_goes_to_zero} and the Markov inequality.
 
 Finally, to prove \eqref{eq:occupation_measure_bounded_rates}, first consider a continuously differentiable function $g\colon\RR^{\Sp\setminus\In}\to\RR$ with compact domain, and let $c_g$ be the maximum of the absolute value of its derivative. We have
 \begin{align*}
  \sup_{t\in[0,T]}E\prq{\int_0^t\abs{g(\widehat{X}^N_s)-g(\widehat{Z}^N_s)}ds}&=\int_0^T E\prq{\abs{g(\widehat{X}^N_s)-g(\widehat{Z}^N_s)}}ds\\
  &\leq\int_0^T c_gE\prq{\abs{\widehat{X}^N_s-\widehat{Z}^N_s}}ds\\
  &\leq Tc_g\sup_{t\in[0,T]}E\prq{\abs{\widehat{X}^N_t-\widehat{Z}^N_t}}ds\xrightarrow[N\to\infty]{}0.
 \end{align*}
 Let $f\colon\RR^{\Sp\setminus\In}\to\RR$ be a continuous function with compact domain. By uniformly approximating $f$ by continuously differentiable functions with compact domain, we have
 $$\sup_{t\in[0,T]}E\prq{\int_0^t\abs{f(\widehat{X}^N_s)-f(\widehat{Z}^N_s)}ds}\xrightarrow[N\to\infty]{}0.$$
 By Markov inequality, it follows that for any $\varepsilon>0$
 \begin{equation}\label{eq:integral_fdd_bounded_rates}
    \sup_{t\in[0,T]}P\pr{\abs{\int_0^t\pr{f(\widehat{X}^N_s)-f(\widehat{Z}^N_s)}ds}>\varepsilon}\xrightarrow[N\to\infty]{}0.
 \end{equation}
 
 Consider the occupation measures on $[0,T]\times\RR^{\Sp\setminus\In}$ given by
 $$\Gamma_1^N([t_1,t_2]\times A)=\int_{t_1}^{t_2}\mathbbm{1}_A(\widehat{X}^N_s)ds\quad\text{and}\quad\Gamma_2^N([t_1,t_2]\times A)=\int_{t_1}^{t_2}\mathbbm{1}_A(\widehat{Z}^N_s)ds$$
 for any $0\leq t_1<t_2\leq T$ and any Borel set $A$ of $\RR^{\Sp\setminus\In}$. By \eqref{eq:bound_1_proof_prop} and \citet[Lemma 1.3]{kurtz:averaging}, we have that the sequences of random measures $(\Gamma_1^N)$ and $(\Gamma_2^N)$ are relatively compact with respect to the Prohorov metric. By the continuous mapping theorem \cite[Section 5.4]{hoffmann:probability}, this in turn implies that the sequence of continuous processes
 $$\pr{\int_0^\cdot f(\widehat{X}^{N}_s)ds,\int_0^\cdot f(\widehat{Z}^{N}_s)ds}=\pr{\int_0^\cdot\int_{\RR^{\Sp\setminus\In}} f(x)d\Gamma_1^{N}(s,x),\int_0^\cdot\int_{\RR^{\Sp\setminus\In}} f(x)d\Gamma_2^{N}(s,x)}$$
 is relatively compact with respect to the weak convergence in the product space $D[0,T]\times D[0,T]$, where $D[0,T]$ denotes the usual Skorohod space. In this case it coincides with weak convergence in the uniform topology since the processes are continuous. The following is inspired by a proof of \citet[Lemma A2.1]{kurtz:countable}. Consider a weak limit $(\widehat{X}_\cdot,\widehat{Z}_\cdot)$, which will be a continuous process. By \eqref{eq:integral_fdd_bounded_rates}, we have $\widehat{X}_t=\widehat{Z}_t$ for any $t$, therefore $d_{Sk}(\widehat{X}_\cdot,\widehat{Z}_\cdot)=0$, where $d_{Sk}$ denotes the Skorohod distance. By the continuous mapping theorem, we have that for any subsequence converging to $(\widehat{X}_\cdot,\widehat{Z}_\cdot)$,
 $$d_{Sk}\pr{\int_0^\cdot f(\widehat{X}^{N_m}_s)ds,\int_0^\cdot f(\widehat{Z}^{N_m}_s)ds}$$
 converges weakly to zero. In particular, this implies that for every $\varepsilon>0$
 $$P\pr{\sup_{t\in[0,T]}\abs{\int_0^t\pr{f(\widehat{X}^{N_m}_s)-f(\widehat{Z}^{N_m}_s)}ds}>\varepsilon}\xrightarrow[m\rightarrow\infty]{}0.$$
 Since the same holds for any convergent subsequence and by relative compactness, \eqref{eq:occupation_measure_bounded_rates} follows for any continuous $f$ with compact support. Indeed, if it were not the case we would have a subsequence such that for some constant $c>0$
 \[P\pr{\sup_{t\in[0,T]}\abs{\int_0^t\pr{f(\widehat{X}^{N_m}_s)-f(\widehat{Z}^{N_m}_s)}ds}>\varepsilon}>c.\]
 However, the subsequence would not contain any convergent subsequence.
 
 Now, let the support of $f$ be not compact, and consider $\nu>0$. There exists a continuous function $f_\nu$ with compact support such that $f_\nu(x)=f(x)$ if $\norm{x}_\infty\leq \Upsilon_{\nu}^{\prime\prime}$, where $\Upsilon_{\nu}^{\prime\prime}$ is as in \eqref{eq:bound_1_proof_prop}. Therefore, if $N$ is large enough
 \begin{multline*}
  P\pr{\sup_{t\in[0,T]}\abs{\int_0^t\pr{f(\widehat{X}^N_s)-f(\widehat{Z}^N_s)}ds}>\varepsilon}\\
  \leq P\pr{\sup_{t\in[0,T]}\abs{\int_0^t\pr{f_\nu(\widehat{X}^N_s)-f_\nu(\widehat{Z}^N_s)}ds}>\varepsilon}+3\nu\xrightarrow[N\rightarrow\infty]{}3\nu.
 \end{multline*}
 The proof is concluded by the arbitrariness of $\nu$.
 \end{proof}
 
 \begin{proof}[Proof of Theorem \ref{thm:weak_convergence_bounded_rates}]
  Fix $\varepsilon>0$. Let $\delta>0$ be such 
  \begin{equation}\label{eq:def_delta}
   \pr{\sum_{r\in\R_S^1}\norm{\xi_r}_\infty B_r+\sum_{i\in\sr}\sum_{j\in\fp}\norm{y_j+y_i}_\infty\sum_{\ell\in\iii}c_{\ell j}B_{i\ell}}\delta<\frac{\varepsilon}{3\theta},
  \end{equation}
  where $c_{\ell j}$ is as in \eqref{eq:limsup_limited_hypothesis}. Now consider a sequence of real numbers $t_0<t_1<t_2<\dots<t_q$ such that $t_{m+1}-t_{m}<\delta$ for any $0\leq m< q$, $t_0=0$ and $t_q=T$. For any $0\leq m< q$ and any species $S\in\Sp\setminus\In$ we have
  \begin{multline*}
   \sup_{t \in[t_{m},t_{m+1}]}\abs{\widehat{Z}_t^N(S)-\widehat{Z}_{t_{m}}^N(S)}\leq N^{-\alpha(S)}\sum_{r\in\R_S^1}\abs{\xi_r(S)}\abs{Y_r(N^{\beta_r}B_r t_{m+1})-Y_r(N^{\beta_r}B_r t_{m})}+\\
   +N^{-\alpha(S)}\sum_{i\in\sr}\sum_{j\in\fp}\abs{y_j(S)-y_i(S)}\sum_{\ell\in\iii}\abs{Y_{i\ell j}\pr{\pi^N_{\ell j}N^{\beta_{i\ell}}B_{i\ell}t_{m+1}}-Y_{i\ell j}\pr{\pi^N_{\ell j}N^{\beta_{i\ell}}B_{i\ell}t_{m}}}.
  \end{multline*}
  The latter is distributed as
  \begin{multline*}
   N^{-\alpha(S)}\sum_{r\in\R_S^1}\abs{\xi_r(S)}Y_r\pr{N^{\beta_r}B_r (t_{m+1}-t_{m})}+\\
   +N^{-\alpha(S)}\sum_{i\in\sr}\sum_{j\in\fp}\abs{y_j(S)-y_i(S)}\sum_{\ell\in\iii}Y_{i\ell j}\pr{\pi^N_{\ell j}N^{\beta_{i\ell}}B_{i\ell}(t_{m+1}-t_{m})},
  \end{multline*}
  which, due to $\alpha(S)>0$, \eqref{eq:limsup_limited_hypothesis}, \eqref{eq:single_scale}, the Law of Large Numbers for Poisson processes and \eqref{eq:def_delta}, is asymptotically bounded in probability by a constant strictly smaller than $\varepsilon/(3\theta)$. In particular,
  $$P\pr{\sup_{t \in[t_{m},t_{m+1}]}\abs{\widehat{Z}_t^N-\widehat{Z}_{t_{m}}^N}>\frac{\varepsilon}{3}}\xrightarrow[N\to\infty]{}0.$$
  Similarly, by \eqref{eq:X_t^N_hat}
  \begin{align*}
   \sup_{t \in[t_{m},t_{m+1}]}\abs{\widehat{X}_t^N(S)-\widehat{X}_{t_{m}}^N(S)}&\leq N^{-\alpha(S)}\sum_{r\in\R_S^1}\abs{\xi_r(S)}\abs{Y_r(N^{\beta_r}B_r t_{m+1})-Y_r(N^{\beta_r}B_r t_{m})}+\\
   &\hspace*{-125pt}\quad+N^{-\alpha(S)}\sum_{i\in\sr}\sum_{j\in\fp}\pr{y_j(S)\sum_{\ell\in\iii}\abs{\overline{M}_{i\ell j}^N(t_{m+1})-\overline{M}_{i\ell j}^N(t_{m})}\quad+y_i\sum_{\ell\in\iii} \abs{M_{i\ell j}^N(t_{m+1})-M_{i\ell j}^N(t_{m})}}\\
   &\hspace*{-125pt}\leq N^{-\alpha(S)}\sum_{r\in\R_S^1}\abs{\xi_r(S)}\abs{Y_r(N^{\beta_r}B_r t_{m+1})-Y_r(N^{\beta_r}B_r t_{m})}+\\
   &\hspace*{-125pt}\quad+N^{-\alpha(S)}\sum_{i\in\sr}\sum_{j\in\fp}(y_j(S)+y_i(S))\sum_{\ell\in\iii}\abs{Y_{i\ell j}\pr{\pi^N_{\ell j}N^{\beta_{i\ell}}B_{i\ell}t_{m+1}}-Y_{i\ell j}\pr{\pi^N_{\ell j}N^{\beta_{i\ell}}B_{i\ell}t_{m}}}+\\
   &\hspace*{-125pt}\quad+y_j(S)\sum_{\ell\in\iii} \pr{\abs{\overline{M}_{i\ell j}^N(t_{m+1})-M_{i\ell j}^N(t_{m+1})}+\abs{\overline{M}_{i\ell j}^N(t_m)-M_{i\ell j}^N(t_{m})}}.
  \end{align*}
  Again, due to $\alpha(S)>0$, \eqref{eq:limsup_limited_hypothesis}, \eqref{eq:single_scale}, the Law of Large Numbers for Poisson processes, \eqref{eq:def_delta} and Lemma \ref{lem:difference_M}, the latter is asymptotically bounded in probability by a constant strictly smaller than $\varepsilon/(3\theta)$. Specifically,
  $$P\pr{\sup_{t \in[t_{m},t_{m+1}]}\abs{\widehat{X}_t^N-\widehat{X}_{t_{m}}^N}>\frac{\varepsilon}{3}}\xrightarrow[N\to\infty]{}0.$$
  We have
  \begin{multline}
   P\pr{\sup_{t\in[0,T]}\abs{\widehat{X}^N_{t}-\widehat{Z}^N_{t}}>\varepsilon}\\
   \leq P\pr{\max_{0\leq m<q}\pr{\widehat{X}^N_{t_m}-\widehat{Z}^N_{t_m},\sup_{t \in[t_{m},t_{m+1}]}\abs{\widehat{Z}_t^N-\widehat{Z}_{t_{m}}^N},\sup_{t \in[t_{m},t_{m+1}]}\abs{\widehat{X}_t^N-\widehat{X}_{t_{m}}^N}}>\frac{\varepsilon}{3}}.
  \end{multline}
  Hence, the proof is concluded by Corollary \ref{cor:cor_to_bounded_rate_lemma}, which is direct consequence of Theorem \ref{thm:eq_true_for_bounded_rates}.
 \end{proof}
 
 \section*{Acknowledgements}
 
 We thank Elisenda Feliu for reading and commenting on early versions of this manuscript and an anonymous reviewer for providing valuable suggestions that lead to improvement of the results. 
 
 \bibliography{../bibliography}

\end{document}